\documentclass[11pt]{article}
\usepackage[margin=1in]{geometry}
 \usepackage{hyperref}
\usepackage{amsmath}
\usepackage{amssymb}
\usepackage{enumerate}

\usepackage[round]{natbib}
\newcounter{savenumi}

\newtheorem{theoremfoo}{Theorem}
\newenvironment{theorem}{\pagebreak[1]\begin{theoremfoo}}{\end{theoremfoo}}

\newenvironment{repeatedtheorem}[1]{\vskip 6pt
\noindent
{\bf Theorem #1}\ \em
}{}

\newtheorem{propositionfoo}{Proposition}  

\newenvironment{proposition}{\pagebreak[1]\begin{propositionfoo}}{\end{propositionfoo}}
\newtheorem{lemmafoo}{Lemma}  

\newenvironment{lemma}{\pagebreak[1]\begin{lemmafoo}}{\end{lemmafoo}}
\newtheorem{claimfoo}[theoremfoo]{Claim}

\newtheorem{observationfoo}[theoremfoo]{Obervation}

\newtheorem{conjecturefoo}[theoremfoo]{Conjecture}

\newtheorem{corollaryfoo}[theoremfoo]{Corollary}

\newtheorem{exercisefoo}{Exercise}

\newtheorem{openfoo}[theoremfoo]{Question}

\newtheorem{nttn}[theoremfoo]{Notation}

\newtheorem{dfntn}{Definition} 
\newenvironment{definition}{\pagebreak[1]\begin{dfntn}\rm}{\end{dfntn}}

\newenvironment{proof}
    {\pagebreak[1]{\narrower\noindent {\bf Proof:\quad\nopagebreak}}}{\QED}

\newcommand{\floor}[1]{\left\lfloor#1\right\rfloor}
\newcommand{\ceiling}[1]{\left\lceil#1\right\rceil}
\def\nre.{$n$\/-r.e.}

\newtheorem{factfoo}[theoremfoo]{Fact}

\newtheorem{propertyfoo}[theoremfoo]{Property}

\makeatletter

\def\@makechapterhead#1{ \vspace*{50pt} { \parindent 0pt \raggedright 
 \ifnum \c@secnumdepth >\m@ne \huge\bf \@chapapp{} \thechapter. \par 
 \vskip 20pt \fi \Huge \bf #1\par 
 \nobreak \vskip 40pt } }

\def\@sect#1#2#3#4#5#6[#7]#8{\ifnum #2>\c@secnumdepth
     \def\@svsec{}\else 
     \refstepcounter{#1}\edef\@svsec{\csname the#1\endcsname.\hskip 1em }\fi
     \@tempskipa #5\relax
      \ifdim \@tempskipa>\z@ 
        \begingroup #6\relax
          \@hangfrom{\hskip #3\relax\@svsec}{\interlinepenalty \@M #8\par}
        \endgroup
       \csname #1mark\endcsname{#7}\addcontentsline
         {toc}{#1}{\ifnum #2>\c@secnumdepth \else
                      \protect\numberline{\csname the#1\endcsname}\fi
                    #7}\else
        \def\@svsechd{#6\hskip #3\@svsec #8\csname #1mark\endcsname
                      {#7}\addcontentsline
                           {toc}{#1}{\ifnum #2>\c@secnumdepth \else
                             \protect\numberline{\csname the#1\endcsname}\fi
                       #7}}\fi
     \@xsect{#5}}

\def\@begintheorem#1#2{\it \trivlist \item[\hskip \labelsep{\bf #1\ #2.}]}

\def\@opargbegintheorem#1#2#3{\it \trivlist
      \item[\hskip \labelsep{\bf #1\ #2\ (#3).}]}

\makeatother

\newif\ifshortconferences
\shortconferencesfalse
\newif\ifmediumconferences
\mediumconferencesfalse

\def\ending#1{{\count1=#1\relax
\count2=\count1
\divide\count2 by 100
\multiply\count2 by 100
\advance\count1 by -\count2
\ifnum\count1=11
th%
\else \ifnum\count1=12
th%
\else \ifnum\count1=13
th%
\else 
\count2=\count1
\divide\count1 by 10
\multiply\count1 by 10
\advance\count2 by -\count1
\ifnum\count2=1
st%
\else \ifnum\count2=2
nd%
\else \ifnum\count2=3
rd%
\else th%
\fi\fi\fi\fi\fi\fi
}}

\def\Proceedingsofthe{\ifshortconferences Proc.\else\ifmediumconferences Proc.\else Proceedings of the\fi\fi}

\newcounter{confnum}

\def\conf#1#2{%
\setcounter{confnum}{#2}%
\addtocounter{confnum}{-\csname #1zero\endcsname}%
\ifnum\value{confnum}=1%
\expandafter\ifx\csname #1One\endcsname\relax%
\Proceedingsofthe\ \arabic{confnum}\ending{\value{confnum}}\ \csname #1name\endcsname%
\else \csname #1One\endcsname\fi%
\else%
\Proceedingsofthe\
\arabic{confnum}\ending{\value{confnum}}\ \csname #1name\endcsname\fi}

\def\qsym{\vrule width0.7ex height0.9em depth0ex}

\newif\ifqed\qedtrue

\def\noqed{\global\qedfalse}

\def\qed{\ifqed{\penalty1000\unskip\nobreak\hfil\penalty50
\hskip2em\hbox{}\nobreak\hfil\qsym
\parfillskip=0pt \finalhyphendemerits=0\par\medskip}\fi\global\qedtrue}

\makeatletter
\def\eqnqed{\noqed
	\def\@tempa{equation}
	\ifx\@tempa\@currenvir\def\@eqnnum{\qsym}%
	\addtocounter{equation}{-1}\else%
    \def\@@eqncr{\let\@tempa\relax
    \ifcase\@eqcnt \def\@tempa{& & &}\or \def\@tempa{& &}%
      \else \def\@tempa{&}\fi
     \@tempa {\def\@eqnnum{{\qsym}}\@eqnnum}
     \global\@eqnswtrue\global\@eqcnt\z@\cr}\fi}

\def\eqnlabel#1#2{\if@filesw {\let\thepage\relax%
   \def\protect{\noexpand\noexpand\noexpand}%
   \edef\@tempa{\write\@auxout{\string
      \newlabel{#2}{{{#1}}{\thepage}}}}%
   \expandafter}\@tempa%
   \if@nobreak \ifvmode\nobreak\fi\fi\fi%
	\def\@tempa{equation}
	\ifx\@tempa\@currenvir\def\theequation{{#1}}%
	\addtocounter{equation}{-1}\else%
    \def\@@eqncr{\let\@tempa\relax
    \ifcase\@eqcnt \def\@tempa{& & &}\or \def\@tempa{& &}%
      \else \def\@tempa{&}\fi
     \@tempa {\def\@eqnnum{{#1}}\@eqnnum}
     \global\@eqnswtrue\global\@eqcnt\z@\cr}\fi}

\makeatother

\def\QED{\qed}

\makeatother

\begin{document}

\date{}
\title{Minimizing Makespan in Sublinear Time via Weighted Random Sampling
}

  \author{ Bin Fu$^1$, Yumei Huo$^2$, Hairong Zhao$^3$
   \\\\
  $^{1}$
  Department of Computer Science\\
  University of Texas Rio Grande Valley, Edinburg, TX 78539, USA\\
  bin.fu@utrgv.edu\\\\
  $^{2}$
  Department of Computer Science\\
  College of Staten Island, CUNY, Staten Island, New York 10314, USA\\
  yumei.huo@csi.cuny.edu
  \\\\
  $^{3}$
  Department of  Computer Science\\
  Purdue University Northwest,  Hammond, IN 46323, USA\\
  hairong@purdue.edu }
  
\maketitle

\begin{abstract} 
We consider the classical makespan minimization scheduling problem where $n$ jobs must be scheduled on $m$ identical machines. 
Using  weighted random sampling, we develop two sublinear time approximation schemes: one for the case where $n$ is known and one for the case where $n$ is unknown. Both algorithms not only give a $(1+\epsilon)$-approximation to the optimal makespan but also generate a sketch schedule. 

Our first algorithm,  
which targets the case where $n$ is known and draws samples  in a single round under weighted random sampling, has a running time of $\tilde{O}(\tfrac{m^5}{\epsilon^4} \sqrt{n}+A(\ceiling{\tfrac{m}{\epsilon}}, m,{\epsilon} ))$, where 
 $A(\mathcal{N}, m, \alpha)$ is the time complexity of any $(1+\alpha)$-approximation scheme for the  makespan minimization of $\mathcal{N}$ jobs on $m$ machines. 
 The second algorithm addresses  the case where $n$ is unknown. It uses adaptive weighted random sampling, 
 and   also  runs in time $\tilde{O}\left( \tfrac{m^5} {\epsilon^4} \sqrt{n}  + 
   A(\ceiling{\tfrac{m}{\epsilon}}, m, {\epsilon} )\right)$.   
\end{abstract}
\textbf{keywords:} sublinear time algorithms, makespan minimization, weighted random sampling

\section{Introduction}
We study the classical makespan minimization scheduling problem on parallel machines, where there are $n$ jobs to be scheduled on $m$ identical parallel machines. The jobs are all available for processing at time $0$ and are labeled $1, 2, \dots, n$.  Each job $j$, $1 \le j \le n$, has a processing time $p_j$ and can be scheduled on any machine. At any time, only one job can be scheduled on each machine. A job cannot be interrupted once it is started. Given a schedule $S$, let $C_j$ be the completion time of job $j$ in $S$,  the makespan of the schedule $S$ is $C_{max} = \max_{1 \le j \le n} C_j$. The objective is to find  the minimum makespan among all schedules, denoted as  $OPT(I)$, or simply $OPT$ if $I$ is clear from the context.   This classical scheduling problem, also  known as the load balancing problem, has many applications in the manufacturing and service industries. Additionally, it plays significant roles in computer science, including applications in client-server networks, database systems, and cloud computing. 

In this paper, our goal is to design sublinear time approximation algorithms for this classical scheduling problem. 
Sublinear time algorithms are suited  for  situations where the input data is so large that even a single scan of the entire input  cannot be afforded, so the solution is given by reading only a minuscule fraction of the input data and using only sublinear time.
As \cite{Czumaj2010} 
noted in their survey paper, sublinear time algorithms have their roots in the study of massive data sets, where even linear-time algorithms can be prohibitively slow.   Hence, there is a need to develop algorithms whose running times are not only polynomial, but actually are sublinear in the input size.  Such algorithms do not need to read the entire input; instead, they determine the output by inspecting only a subset of input elements.
Sublinear time algorithms are typically randomized  and   return only an approximate solution rather than the exact one. Developing sublinear time algorithms not only speeds up computation, but also reveals some interesting properties of computation, especially the power of randomization. 

So far, most sublinear time algorithms have been developed based on uniform random sampling. While uniform sampling is easy to implement, it has a drawback in scheduling problems: When the jobs' processing times vary significantly, 
failing to sample the largest processing time may result in a significant loss of accuracy. 
Consider two instances of job inputs: in the first input, all job processing times are a small number $\epsilon$, while in the second input, all jobs' processing times are $\epsilon$ except for one job, which has a very large processing time $X$. No sampling scheme can distinguish between these two inputs until it samples the job with processing time 
$X$, and with uniform sampling, it would require a linear number of samples to hit $X$.
Therefore, in this paper, we develop sublinear time algorithms using weighted random sampling. We formally define weighted random sampling as follows.

\begin{definition}
Given $n$ elements, where each element $j$ is associated with a weight $w_j$, and the total weight is denoted by $W = \sum_{j=1}^n w_j$,
{\bf Weighted Random Sampling} (WRS) is a technique in which the probability of selecting an element is proportional to its weight, i.e.,  $\tfrac{w_j}{W}$. 
\end{definition}

Our weighted sampling model  is consistent with prior work on sublinear algorithms. 
 \cite{MPX07} used weighted sampling for estimating the sum of a set. \cite{lj22} studied the same problem, also employing weighted sampling, and achieved improved complexity. \cite{BP25} studied the moment estimation problem using weighted sampling. These works treat weighted sampling as a basic operation and focus on understanding the number of samples required to achieve a target accuracy. Our results contribute to this line of research by applying the same sampling model to scheduling.   
 
Weighted sampling arises naturally in several domains, including router-level packet monitoring (e.g., NetFlow, sFlow) and pre-indexed database systems. It can be implemented in various ways, with trade-offs depending on the data-access model, update pattern, and performance goals. Prior work includes Markov chain–based methods, streaming algorithms such as weighted reservoir sampling, preprocessing-based techniques for static data, and rejection sampling (\citealp{Metropolis53, ES06, Walker77, Mel24}). It remains an open question whether weighted sampling can itself be constructed in sublinear time for arbitrary  inputs. Exploring this possibility lies beyond the scope of this paper but represents a natural and complementary direction for future investigation.

\subsection{Related work}

The makespan minimization problem has been known to be NP-hard for a long time. In the 1960s, Graham showed that the List Scheduling rule and the Longest Processing Time First rule give approximations of  $2-\tfrac{1}{m}$ and $\tfrac{4}{3}-\tfrac{1}{3m}$, respectively (see \citealp{g66, g69}). 
Later, many approximation schemes have been developed for both the case where the number of machines $m$ is fixed 
(\citealp{hs76}) and the case where $m$ is arbitrary 
(\citealp{hds87,Jansen10,JansenKleinVerschae20}).  

Over the past few decades, sublinear time algorithms have been developed for a variety of applications, including algebra, graph theory, geometry, string and set operations, optimization, and probability theory (\citealp{ChazelleLiuMagen05, Feige06, GoldreichRon06, ChazelleRubinfeldTrevisan05, CzumajSohler04, CzumajErgun05, FuChen06, GoldreichGoldwasserRon96, GoldreichRon00, Behnezhad2022, Behnezhad20232, Behnezhad20231}). Most of these algorithms rely on uniform random sampling, with a few using weighted sampling (\citealp{MPX07, lj22, BATU20095082}).
In particular, \cite{BATU20095082} 
developed asymptotic approximation scheme  for bin packing problem that runs in $\tilde{O}(\sqrt{n} \cdot poly(1/
\epsilon) + g(1/\epsilon)$ using weighted sampling alone, where $g(x)$ is an exponential function of $x$ and in  $\tilde{O}( n^{1/3} \cdot poly(1/
\epsilon) + g(1/\epsilon)$     
time when both weighted sampling and uniform sampling are allowed.   In addition to an approximate value,
the algorithm can also output a
constant-size ``template'' of a packing that can later be used to find a near-optimal packing
in linear time.

For the scheduling problem, to the best of our knowledge, the only works on  sublinear time algorithms are those by  \cite{fhz24-chain, fhz25sublinear}   
for the classical scheduling problems under uniform random sampling setting. In \cite{fhz24-chain}, sublinear time algorithms are developed for the parallel machine makespan minimization scheduling problem with chain precedence constraints and in \cite{fhz25sublinear} the sublinear time algorithms are presented for the classical parallel machine makespan minimization problem, and the algorithm is then extended to the problem with precedence constraints where the precedence graph has limited depth.  Both works use uniform random sampling and assume constrained input processing times. 

In this paper, we  consider the classical makespan minimization problem where the processing times are arbitrary, and develop  sublinear time approximation schemes using weighted sampling instead of uniform sampling.
 A brief discussion in~\cite{BATU20095082} suggests that their sketch method for bin packing could be applied to approximate makespan minimization problem, but no detailed analysis is presented. 
Moreover, their method 
only works for the case that  $n$ is known. In this paper, we develop two algorithms,  
one for the case in which $n$ is known, another, more involved algorithm,  for the case where $n$ is unknown.  Furthermore, as in   \cite{BATU20095082}, 
our algorithms can also be adapted to output not only an approximate value, but also a sketch of scheduling that takes a sublinear space and later can be used to generate a near optimal schedule in linear time. Such output is essential for many applications that require not only an approximate optimal value, but also a schedule specifying the assignment of   jobs to machines.    

\subsection{New contribution}

In this paper, we develop sublinear time algorithms for the makespan minimization problem using weighted random sampling where the weight of a job is simply its processing time. Note that the  lower bound of $\Omega(\sqrt{n})$ by~\cite{MPX07} for computing an approximate sum of $n$ positive numbers using weighted random sampling also applies to our problem.  

Our sublinear algorithms start with a key step, which is to derive a sketch of the input instance by drawing $\tilde{O}(\sqrt{n})$ samples.
When the number of jobs $n$ is given, we draw 
the samples in a single round. With high probability, the samples will include all jobs with large processing time and some medium jobs, but no small jobs. To construct the sketch of the input, we include all the sampled large jobs, divide the  medium jobs into groups, where the number of jobs in each group  is estimated using a generalized birthday paradox argument based on the samples from the group, and disregard the small jobs.   
When the number of jobs $n$ is unknown, we employ adaptive sampling to estimate the number of jobs. Specifically, we sample jobs in several rounds, increasing the number of samples geometrically until we can apply the birthday paradox argument to estimate the number of jobs in all groups.  

The second step of our algorithms is to compute an  approximation of the optimal makespan for the scheduling  problem based on the sketch constructed from the first step. 
To this end, we may employ any existing approximation scheme for makespan minimization (\citealp{hochbaum1997scheduling,
Jansen10, JansenKleinVerschae20}) as a black-box procedure and apply it to the large jobs.

In many applications, one needs not only an approximate value of the optimal makespan, but also a schedule that specifies how the jobs are assigned to machines. To meet this requirement, we use the concept of  ``sketch schedule'' and show that we can modify our algorithms to compute ``sketch schedule'' while keeping the running time still sublinear in the number of jobs. Moreover, we present how ``sketch schedule'' can be used to construct  a concrete schedule for jobs in $I$ when full job information is available. 

Let $A(\mathcal{N}, m, \alpha)$ be the time complexity of any $(1+\alpha)$-approximation scheme  for scheduling $\mathcal{N}$ jobs on $m$ machines. 
Our main results can be summarized in the following two theorems.

\begin{theorem}  \label{theorem:sublinear-n-known}  When the number of jobs $n$ is known, 
using non-adaptive weighted sampling, there is a randomized $(1+\epsilon)$-approximation algorithm for the   makespan minimization problem  that runs in $\tilde{O}(\tfrac{m^5}{\epsilon^4} \sqrt{n}  +  A(\ceiling{\tfrac{m} {\epsilon}}, m, {\epsilon} ) )$ time. Furthermore, it can compute a sketch schedule $\tilde{S}$, represented using   
$O(\tfrac{m^2}{ \epsilon^2} (\log \tfrac{n m}{\epsilon}))$ space,
which can be used to generate a schedule  of the jobs subsequently  with a  makespan of
 \end{theorem}

\begin{theorem} \label{theorem:sublinear-n-unknown}   
When the number of jobs is unknown, there is a randomized $(1+ \epsilon)$-approximation algorithm for the makespan scheduling problem  that uses adaptive weighted random samplings,   runs in $\tilde{O}\left( \tfrac{m^5} {\epsilon^4} \sqrt{n}    +  A(\ceiling{\tfrac{m} {\epsilon}}, m, {\epsilon} ) \right)   $ 
 time. 
Furthermore, it can compute  a sketch schedule, represented using   $O(\tfrac{m^2}{\epsilon^2} (\log {\tfrac{nm}{ \epsilon}))}$ space, which can be used later to generate a schedule with a makespan at most $(1+ 3\epsilon)$ times the optimal.  
\end{theorem}

\medskip
The paper is organized as follows. In Section~\ref{frame-sec}, we introduce some key definitions and 
present the structure of our sublinear time algorithms, which consist of two main steps: constructing a sketch of the input and computing an approximate value of the optimal makespan based on the sketch. In this section, we provide a general analysis of the algorithm, and then in Section~\ref{subsec:approx-for-sketch}, we derive the sublinear time approximation schemes to compute the approximate value of the optimal makespan, assuming the sketch of the input has been derived. The detailed algorithms for constructing the sketch via weighted sampling are presented and analyzed in Section~\ref{derive-sketch}, where we distinguish between the cases in which the number of jobs, $n$, is known (Section~\ref{subsec: Sublinear1-sec}) and unknown (Section~\ref{subsec: Sublinear2-sec}).
Then in Section~\ref{Discussion:Deterministic-Algorithm-Sec}, we discuss how to apply the idea of input sketching to derive a deterministic approximation algorithm for our scheduling problems. 
Finally, we draw the concluding remarks in Section~\ref{sumary-sec}.

\section{Definitions and Framework of Our sublinear Time Algorithms}\label{frame-sec}

\subsection{Definitions}

Given an instance of the problem $I = \{ p_j, 1 \le j \le n\}$, our algorithms  draw a sublinear size of samples, then estimate the original job set based on the samples, and finally compute the approximate value of the optimal makespan based on the estimated job set. To estimate instance $I$, we construct a ``sketch'' instance  $\tilde{I}$, which can be stored in a space of size sublinear in $n$. 
Let $p_{max}$ be the largest job processing time in the input. Since $p_{max}$ is  unknown,  our algorithm will first obtain an upper bound of $p_{max}$,  denoted as $p_{max}'$. For a parameter $\delta$, $ 0 < \delta < 1$,  which will be decided later, we   partition the jobs into groups such that group $k$ contains all the jobs with the processing times in the interval  $I_k = (p_{max}'(1-\delta)^{k}, p_{max}'(1-\delta)^{k-1}]$, $k \ge 1$. For convenience, we also use $j \in I_k$   interchangeably with $p_j \in I_k$. A ``sketch'' instance  $\tilde{I}$ is formally defined as follows.

 \begin{definition}
 Given an instance of the problem $I = \{ p_j, 1 \le j \le n\}$ and $p_{max}'$  such that $p_{max}\le p_{max}'$,  let $n_k$ be the number of jobs $j \in I_k$ where $I_k = (p_{max}'(1-\delta)^{k}, p_{max}'(1-\delta)^{k-1}]$, $k \ge 1$. An input sketch of $I$ is denoted as $\tilde{I} = \{\langle  \tilde{n}_k, \tilde{p}_k\rangle\}$, where $\langle \tilde{n}_k, \tilde{p}_k\rangle$ provides an estimate of jobs of $I$ from  the interval $I_k$: $\tilde{n}_k$ is an estimate of the numbers of the jobs, $n_k$,  and $\tilde{p}_k =  p_{max}'(1-\delta)^{k-1}$ is the estimated processing time for these jobs.   
\end{definition}

Next, we  define an $(\alpha,\beta_1,\beta_2)$-sketch of $I$, which characterize how well a sketch approximates $I$   in terms of the number of jobs in each interval,  and $OPT(I)$. Without loss of generality, the sketch only contains those intervals $I_k$  for which $\tilde{n}_k > 0$. 

 \begin{definition}
 Given an instance of the problem $I$ 
 we say an input sketch  $\tilde{I} = \{\langle  \tilde{n}_k, \tilde{p}_k\rangle: 1 \le k \le t \}$ is an $(\alpha,\beta_1,\beta_2)$-sketch of $I$  if the following conditions hold:
\begin{enumerate}
    \item     For every $k \in [1,t]$ with $\tilde{n}_k > 0$,
    $   (1-\alpha)\, n_k \le \tilde{n}_k \le (1+\alpha)\, n_k $

    \item 
    $
    (1-\beta_1)\, OPT(I) \le OPT(\tilde{I}) \le (1+\beta_1)\, OPT(I).
    $
    \item The total processing time of the jobs from intervals $I_k$ not included in the sketch, i.e.,  those with $\tilde{n}_k = 0$ or the intervals $I_k, k>t$,   is at most $\beta_2 \cdot OPT(I)$. 
   \end{enumerate}
\end{definition}

\subsection{The Framework of the Sublinear Time Algorithms}
Our sublinear time algorithms follow the same structure, which consists of two stages: compute a sketch and derive an approximation from the sketch instance. Formally, the main structure is presented as follows.

 \medskip 
{\noindent \bf Main Algorithm:}

Input: $\epsilon$, $m$,  $ I = \{ p_j: 1 \le j \le n \}$
 
Output: an approximate value of $OPT(I)$ 

\begin{enumerate}
     \item generate an $(\alpha,\beta_1,\beta_2)$-sketch instance for the original instance $I$,  
     $ \tilde {I} =\{\langle \tilde {n}_k, \tilde {p}_k\rangle , 1 \le k \le t \}$, where $\alpha$, $\beta_1$ and $\beta_2$ all depend on $\epsilon$  
    \item compute an approximate value of $OPT(I)$ based on $\tilde{I}$.
\end{enumerate}

If there exist efficient algorithms $\mathcal{A}_1$ and $\mathcal{A}_2$ for computing  a sketch instance and  an approximation for the sketch instance, respectively, we can simply apply $\mathcal{A}_1$ in the first step 
and $\mathcal{A}_2$ in the  second step, and thus obtain an approximation for the problem, as presented in the following proposition.

\begin{proposition}\label{proposition}
Assuming that (1) there exists an algorithm $\mathcal{A}_1(.)$ that computes  an $(\alpha,\beta_1, \beta_2)$-sketch for any given instance $I$ of $n$ jobs in time $T_1(\alpha, \beta_1,\beta_2,n,m)$, and (2) there exists another algorithm $\mathcal{A}_2(.)$ that computes a $(1+\gamma)$-approximation for a sketch instance $\tilde{I} = \{\langle  \tilde{n}_k, \tilde{p}_k\rangle: 1 \le k \le t \}$,  $\tilde{n}  = \sum_{k=1}^t \tilde{n}_k$,  in $T_2(\gamma, t, \tilde{n},m)$ time. 
 Then there exists an approximation algorithm for any instance $I$ that returns a value $X$, such that $ (1- \beta_1 ) OPT(I) \le X \le  (1+\gamma)(1+\beta_1) OPT(I)$ with running time $O(T_1(\alpha, \beta_1,\beta_2,n,m)+T_2(\gamma,t, (1+\alpha) n,m))$.
\end{proposition}

So far in the literature,  no existing algorithm has been develped to obtain an input sketch  $\tilde{I} = \{\langle  \tilde{n}_k, \tilde{p}_k\rangle\}$, and more specifically, an $(\alpha,\beta_1,\beta_2)$-sketch of $I$ in sublinear time. 
Therefore, one of the main tasks in this paper is to develop sublinear time  algorithms for obtaining an $(\alpha,\beta_1,\beta_2)$-sketch for any instance $I$. The details of these algorithms will be presented in Section~\ref{derive-sketch}.

As for the second step, although there is no existing sublinear time algorithm for computing an  approximation from  the input sketch, we may employ the existing approximation schemes that have been developed for the classical makespan minimization scheduling problem and develop sublinear time algorithms for the input sketch. We present the details in the following subsection.

\subsection{Sublinear Time Approximation Schemes for a Sketch Instance}\label{subsec:approx-for-sketch}  

Let $A$ be any approximation scheme for the classical parallel machine makespan minimization scheduling problem, such as the ones proposed in ~\cite{hochbaum1997scheduling,
Jansen10} and the most recent one developed by Jansen, Klein, and Verschae~\cite{JansenKleinVerschae20}. We design an efficient meta-algorithm that uses $A$ as a black box to obtain $(1+\epsilon)$-approximation for any sketch instance. The main idea is to apply the black box approximation scheme to the largest jobs and use a batch processing method to assign the remaining jobs.

\medskip
{\noindent \bf Meta-Algorithm:} 

\noindent Input: $\epsilon$, $m$, 
$ \tilde {I} =\{\langle \tilde {n}_k, \tilde {p}_k\rangle: 1 \le k \le t, \tilde {p}_k > \tilde {p}_{k+1} \text{ for all } k < t\}$, and $P = \sum_{k=1}^t \tilde {n}_k \tilde {p}_k$.

\hspace{0.2in} $A$: any approximation scheme for classical makespan minimization problem


\noindent Output: a $(1+ \epsilon)$-approximation  of  $OPT(\tilde {I})$

\begin{enumerate}
\item let $\delta ={\tfrac{\epsilon}{3}}$ and $h(m,\delta)=\ceiling{ \tfrac{m}{\delta}}$
\item 
apply $A$ 
to the $h(m, \delta)$ largest jobs from $\tilde {I}$ to obtain a  schedule $S$ whose makespan is at most $(1 + \delta)$ times that of the optimal schedule for these jobs  

 \item let $T_0$ be the makespan of $S$  
 
 \item return $T = ( 1+ \delta) \max(T_0,  P/m)$
 \end{enumerate}

\begin{theorem}\label{trans-thm}
Assume $A$ is a $(1+\alpha)$-approximation scheme for the makespan minimization problem of scheduling $n$ jobs on $m$ machines, with running time $A(n,m,\alpha)$. Given any sketch instance $\tilde {I} = \{ \langle \tilde {n}_1, \tilde {p}_1\rangle, $ $\langle \tilde {n}_2, \tilde {p}_2\rangle,\cdots, \langle \tilde {n}_t, \tilde {p}_t\rangle \}$, the Meta-Algorithm returns a $(1+\epsilon)$-approximation for $\tilde {I}$ in time $O(A(h(m,\delta), m, \delta))$, where $\delta =\tfrac{\epsilon}{3}$, and $h(m,\delta)=\ceiling{\tfrac{m}{\delta}}$. 
\end{theorem}   

\begin{proof} The time complexity is straight forward.  We just show that $ T \le (1+\epsilon) OPT(\tilde{I})$ and there exists a feasible schedule whose makespan is at most $T$. 

First, by assumption, when we apply $A$ to the 
 $h(m, \delta)$ largest jobs, we get $T_0 \le (1 + \delta)OPT(\tilde{I})$. Since $P/m$ is  a lower bound of $OPT(\tilde{I})$,   we have $\max(  P/m,  T_0 ) \le (1 + \delta)OPT(\tilde{I})$ and
 $$T = ( 1+ \delta) \max(  P/m,  T_0 ) \le ( 1+ \delta)^2 OPT(\tilde{I}) \le (1 + \epsilon)OPT(\tilde{I}).$$

Next, we show that there exists a feasible schedule whose makespan is at most $T$.  Let $\tilde{p}_{k_0}$ be the largest processing time of the remaining jobs. Then,  the processing time of  each of the   $h(m, \delta)$ largest jobs is at least  $\tilde{p}_{k_0}$, and  thus the makespan of $S$ is  at least  $ T_0 \ge h(m,\delta) \tilde{p}_{k_0}/m $. We claim that we can schedule all the remaining jobs based on $S$ 
by time $T$ using List Scheduling. 

We consider two cases: (1) $P/m < T_0$.  In this case, $T=(1+\delta)T_0$, and $T  - T_0= \delta T_0  \ge \delta  \cdot h(m,\delta) \cdot \tilde{p}_{k_0} / m    = \delta \cdot \ceiling{\tfrac{m}{\delta}} \cdot \tilde {p}_{k_0} / m \ge \tilde {p}_{k_0}.$
This means that every remaining job can fit in the interval between $T_0$ and $T$.  
If at some point a job cannot be scheduled with the completion time at or before $T$ on any machine, then we must have that all machines are busy at time $T_0 > P/m$. But this is impossible, since the total processing time is 
$P$. (2) $P/m \ge T_0$. In this case, $T = (1 + \delta) P/m$. Then the interval between $T$ and $P/m $ has length  $T- P/m = \delta P/m \ge \delta T_0 \ge \tilde {p}_{k_0}$. In other words, every remaining job can fit in the interval between $P/m$ and $T$. If at some point a job cannot be scheduled with the completion time at or before $T$, we get a contradiction as in the first case.
\end{proof}

Theorem~\ref{trans-thm} shows that the Meta-Algorithm yields a $(1+\epsilon)$-approximation for any sketch instance $\tilde{I}$. Combining with Proposition~\ref{proposition}, it is straightforward to see that applying the Meta-Algorithm to an $(\alpha, \beta_1, \beta_2)$-sketch produces an approximation algorithm for the original instance $I$. We summarize this result as follows.

\begin{theorem}\label{orig-thm}
Assuming that there exists an algorithm $\mathcal{A}_1(.)$ that computes  an $(\alpha,\beta_1, \beta_2)$-sketch for any given instance $I$ of $n$ jobs in time $T_1(\alpha, \beta_1,\beta_2,n,m)$, applying Meta-Algorithm to this sketch obtains an approximation algorithm for the instance $I$ that returns a value $X$, such that $ (1- \beta_1 ) OPT(I) \le X \le  (1+\epsilon)(1+\beta_1) OPT(I)$ with running time $O(T_1(\alpha, \beta_1,\beta_2,n,m)+A(h(m,\delta), m, \delta))$, where $\delta =\tfrac{\epsilon}{3}$, and $h(m,\delta)=\ceiling{\tfrac{m}{\delta}}$.
\end{theorem}

The Meta-algorithm returns only an approximation of the optimal makespan, without providing insight into how to construct a schedule with that makespan. However, in many applications, one needs not only an approximate value of the optimal makespan, but also a schedule that specifies how the jobs are assigned to machines. In the following, we first introduce the concept of  ``sketch schedule'' and show we can modify the Meta-algorithm to compute a ``sketch schedule'' while still keeping the overall running time sublinear. Then, we present how a ``sketch schedule'' can be used to construct a concrete schedule for jobs in $I$ when full job information is available.

\subsection{Sketch Schedule}

\begin{definition} Given a sketch instance $ \tilde {I} =\{\langle \tilde {n}_k, \tilde {p}_k\rangle,  1 \le k \le t \} $ for instance $I$,
    we denote a {\it  sketch schedule} for $I$  as   $\tilde{S}=\{\langle  \tilde{n}_{i,j}, \tilde{p}_j\rangle, 1 \le i \le m, 1 \le j \le t, \sum_i \tilde{n}_{i,j} = \tilde{n}_j\} $, where $\langle \tilde{n}_{i,j}, \tilde{p}_j\rangle$ represents that  $\tilde{n}_{i,j}$ jobs with processing time $\tilde{p}_j$ are assigned to machine $M_i$.  
\end{definition}

Now we show that we can modify the Meta-algorithm to generate  a sketch schedule.  

\medskip
{\noindent \bf Modified Meta-Algorithm:}  

\noindent Input: $\epsilon$, $m$, 
$ \tilde {I} =\{\langle \tilde {n}_k, \tilde {p}_k\rangle: 1 \le k \le t, \tilde {p}_k > \tilde {p}_{k+1}  \text{ for all } k < t\}$, and $P = \sum_{k=1}^t \tilde {n}_k \tilde {p}_k$.

\hspace{0.2in} $A$: any approximation scheme for classical makespan minimization problem


\noindent Output: a $(1+ \epsilon)$-approximation  of  $OPT(\tilde {I})$ and a sketch
  schedule for $\tilde {I}$

\begin{enumerate}
\item let $\delta =\tfrac{\epsilon}{3}$ and $h(m,\delta)=\ceiling{\tfrac{m}{\delta}}$
\item  \label{A-for-big-jobs}
let $S$ be the schedule of the $h(m, \delta)$ largest jobs returned by $A$; 
let $\tilde{S}$ $=\{\langle  \tilde{n}_{i,j}, \tilde{p}_j\rangle, 1 \le i \le m, 1 \le j \le t
\} $ be the sketch schedule  that describes how the $h(m, \delta)$ largest jobs  are assigned to the $m$ machines.

 \item let  $T_i$ denotes the completion time of the last job scheduled on machine $i$ in $S$, and let  $T_0 = \max_{1 \le i \le t} T_i$  
 \item let $T = ( 1+ \delta) \max(T_0,  P/m)$
 \item let $\tilde {n}'_{k}$, $1 \le k \le t$, be the number of jobs  with processing time $\tilde {p}_{k}$ that are not assigned by $A$,  and let  $k_0$ be the smallest index such that $\tilde {n}'_{k_0}>0$ 
 \item \label{assign-remaining-jobs} compute the assignment of these jobs as follows:
  \begin{enumerate}[a.]
 \item $i=1$ 
 \item\label{while-loop-line} while $k_0 \le t$ 
     \begin{enumerate}[i.]
     \item compute the number of jobs from $\langle \tilde {n}'_{k_0}, \tilde {p}_{k_0} \rangle$ that can be scheduled on $M_i$ by time $T$, that is, $\tilde {n}''_{k_0}=\min\{\tilde {n}'_{k_0}, \floor{\tfrac{T - T_i}{\tilde {p}_{k_0}}}\}$
     \item  update $T_i=T_i+\tilde {n}''_{k_0} \tilde {p}_{k_0}$ and $\tilde {n}'_{k_0}=\tilde {n}'_{k_0}-\tilde {n}''_{k_0}$
     \item if $\tilde {n}'_{k_0}=0$, then $k_0=k_0+1$, else  $i=i+1$   
   \end{enumerate}
 
 \end{enumerate}

 \item return $T$ and $\tilde{S}$
\end{enumerate}

\begin{theorem}\label{meta-sketch}
  Given any sketch instance $\tilde {I} = \{ \langle \tilde {n}_1, \tilde {p}_1\rangle, $ $\langle \tilde {n}_2, \tilde {p}_2\rangle,\cdots, \langle \tilde {n}_t, \tilde {p}_t\rangle \}$. A sketch schedule can be computed in $O(A(h(m, \delta), m, \delta)+m+t)$ time and stored using  $O(m\cdot t)$ space. 
\end{theorem}

\begin{proof}
We apply the Modified Meta-algorithm. It takes $A(h(m, \epsilon), m, \delta)$ time  at Step~\ref{A-for-big-jobs}. We need $O(m+t)$ steps to schedule the remaining jobs. This is because, in each iteration of  Step~\ref{assign-remaining-jobs}, 
either $i$ or $k_0$ is incremented until $k_0=t$, with $i$ at most $m$. 
Therefore, the running time is $O(A(h(m, \delta), m, \delta)+m+t)$.
For space complexity, it is obvious that we need 
  $O(m\cdot t)$ space to store the sketch schedule $\tilde{S} =\{\langle  \tilde{n}_{i,j}, \tilde{p}_j\rangle, 1 \le i \le m, 1 \le j \le t\}$. 
\end{proof}

%
Next, we show that the sketch schedule can be used to generate a schedule of all $n$ jobs when all the information of all jobs are accessed.

\begin{theorem}\label{theorem:sketch-to-schedule}
Given  an instance  $I$,   an $(\alpha, \beta_1, \beta_2)$-sketch  $\tilde{I}$  and a sketch schedule $\tilde{S}$, we can  generate a schedule $S$ for instance  $I$  based on $\tilde{S}$ and the  makespan of $S$ is at most $ ((1+\beta_1)(1 + \tfrac{\alpha}{1 - \alpha}  m ) + \beta_2 )OPT(I)$.  
\end{theorem}
 
\begin{proof}
Based on the sketch schedule $\tilde{S}$, we can construct a schedule $S$ for the job instance $I$ interval by interval. 
For each interval $I_k, \tilde{n}_k > 0, 1 \le k \le t$, by the definition of $(\alpha, \beta_1, \beta_2)$-sketch, 
we have $(1-\alpha)\, n_k \le \tilde{n}_k \le (1+\alpha)\, n_k $, that is, 
$$  |( n_k - \tilde{n}_k )| \le \alpha \cdot n_k \le \tfrac{\alpha}{1 - \alpha} \tilde{n}_k. $$
To construct $S$,  we try to schedule the jobs from $I_k$ in the same way as  in $\tilde{S}$. In the case of $n_k > \tilde{n}_{k}$, there will be at most  $\tfrac{\alpha}{1 - \alpha} \tilde{n}_k$ jobs left unscheduled. We schedule these jobs to machine $i$ that has the largest number of jobs from $I_k$ scheduled, that is, $\tilde{n}_{i, k}$ is the largest among all $1 \le i \le m$, which implies $\tilde{n}_{i, k} \ge \tfrac{\tilde{n}_{k}}{m}$. Combining all the intervals, it can be shown that the makespan of $S$ is increased by at most $\tfrac{\alpha}{1 - \alpha} m $ times the makespan of $\tilde{S}$, which in turn is at most  $(1 + \beta_1) OPT(I)$.  
For the intervals $I_k$ with $\tilde{n}_k = 0$ or $k \ge t$, we can schedule the jobs from these intervals on any machine, by definition of $(\alpha, \beta_1, \beta_2)$-sketch, the makespan will increase by at most  $\beta_2 OPT$. In summary, the makespan of the constructed schedule is at most $((1+\beta_1)(1 + \tfrac{\alpha}{1 - \alpha}  m ) + \beta_2 )OPT(I) $.
\end{proof}

To summarize, we have shown that,  for a given instance $I$,  if we can obtain a $(\alpha,\beta_1,\beta_2)$-sketch $\tilde {I} = \{ \langle \tilde {n}_1, \tilde {p}_1\rangle, $ $\langle \tilde {n}_2, \tilde {p}_2\rangle,\cdots, \langle \tilde {n}_t, \tilde {p}_t\rangle \}$ where $t$ is sublinear in $n$, then  we can get both an approximate value of  $OPT(I)$ and a sketch schedule in sublinear time. 
In the following, we will  focus on designing sublinear time algorithms to obtain such an $(\alpha,\beta_1,\beta_2)$-sketch for $I$. We first present some preliminaries. Then, we describe the algorithms in two sections, one for the case when the total number of job $n$ is known, and one for when  $n$ is unknown.

\section{Computing the Sketch Instance in Sublinear Time}\label{derive-sketch}

\subsection{Preliminaries}\label{preliminary-sec}
We provide preliminaries on proportional weighted random sampling and the generalized birthday paradox argument, which are the key tools used in our algorithms and their proofs.  
In addition, we use the well-known Chernoff bounds and union bound, which are given below for convenience.

\medskip
\noindent\textbf{Chernoff Bounds:} 
Let $X_1,\ldots , X_n$ be $n$ independent random $0$-$1$ variables and $X=\sum_{i=1}^n X_i$.
\begin{enumerate}
    \item\label{case1-chernoff} 
     If $X_i$ takes $1$ with probability at most $p$ for $i=1,\ldots , n$, then for  any $  \alpha > 0 $, $\Pr(X>(1+\alpha)pn) < \left[ \tfrac {e^{\alpha} }{ (1+\alpha)^{(1+\alpha)}}\right]^{pn}$; furthermore if $0 < \alpha \le 1$, $ \Pr(X>(1+\alpha)pn) < e^{-\tfrac{1}{3}\alpha^2 pn}$; 
 
 \item If  $X_i$ takes $1$ with probability at least $p$ for $i=1,\ldots , n$, then for any $ 0 < \alpha < 1 $, $\Pr(X<(1-\alpha)pn)<e^{-\tfrac{1}{2}\alpha^2 pn}$. 
 \end{enumerate}
\noindent\textbf{Union Bound:} 
Let $E_1,E_2,\ldots, E_m$ be $m$ events, then we have  $$\Pr(E_1\cup E_2 \ldots \cup E_m)\le \Pr(E_1)+\Pr(E_2)+\ldots+\Pr(E_m).$$

\noindent \textbf{Classical Birthday Paradox:}
Given a group of $n$ people, the probability $pr_n$ that at least two people share the same birthday is:
\[
pr_n = 1 - \prod_{k=0}^{n-1} \left( 1 - \tfrac{k}{365} \right)
\]
 
In the following, we give some properties of weighted random sampling and generalize the classical birthday paradox.  
Given a set of elements with weights $w_1\le w_2\le\ldots\le w_n$, and $W = \sum_{j=1}^n w_j$,
  assume that we draw $k$ samples with replacement using weighted sampling. Let $Pr_{n,k}$ denote the probability that all $k$ samples are distinct. Then we have the following bounds for $Pr_{n,k}$.
 
\begin{lemma}\label{lemma:pr-n-k}
$$\prod_{j = 1}^{k-1} \left( 1 - \tfrac {\sum_{i=1}^{j} { w_{n-i+1}}}{W} \right) \le Pr_{n,k} \le \prod_{j = 1}^{k-1} \left(1-\tfrac{ \sum_{i=1}^{j} {w_i}}{W} \right) \enspace.$$

\end{lemma}
\begin{proof}
We prove 
by induction. When $k=1$, $Pr_{n,k} =1 $; thus the inequalities trivially hold.  Let us assume they  hold for any  $k \ge 1$, 
 then the probability that the $(k+1)$-th sample is different from  the first $k$ samples is $(1-\tfrac{W_k}{W})$ where $W_k$ is the total weight of the first $k$ samples,  and thus $Pr_{n, k+1}$, the probability that all $k+1$ samples are distinct, is given by $Pr_{n, k+1}=Pr_{n,k} \cdot (1-\tfrac{W_k}{W})$.
Since $\sum_{i=1}^{k} w_i \le W_k \le \sum_{i=1}^{k} w_{n-i + 1}$, we have $Pr_{n,k}(1-\tfrac{\sum_{i=1}^{k} w_{n-i + 1}}{W}) \le Pr_{n, k+1} \le Pr_{n,k}(1-\tfrac{\sum_{i=1}^{k} w_i}{W})$.
By induction, 
$$ 
 \prod_{j = 1}^{k-1} \left( 1 - \tfrac {\sum_{i=1}^{j} w_{n-i+1}}{W} \right) \le Pr_{n,k} \le \prod_{j = 1}^{k-1} \left(1-\tfrac{ \sum_{i=1}^{j} {w_i}}{W} \right). $$ 
Therefore,
\begin{eqnarray*}
       Pr_{n, k+1} \ge  \left[ { \prod_{j = 1}^{k-1} \left( 1 - \tfrac {\sum_{i=1}^{j} w_{n-i+1}}{W} \right) }\right ] \cdot \left( 1-\tfrac{\sum_{i=1}^{k} w_{n-i+1}}{W}\right) =  \prod_{j = 1}^{k} \left( 1 - \tfrac {\sum_{i=1}^{j} w_{n-i+1}}{W} \right)   \end{eqnarray*} 
  and 
  \begin{eqnarray*} 
  Pr_{n, k+1} \le \left[\prod_{j = 1}^{k-1} \left(1-\tfrac{ \sum_{i=1}^{j} {w_i}}{W} \right)  \right] \left(1-\tfrac{\sum_{i=1}^{k} {w_i}}{W}\right) =  \prod_{j = 1}^{k} \left(1-\tfrac{ \sum_{i=1}^{j} {w_i}}{W} \right)   
  \enspace.
\end{eqnarray*} 
\end{proof}

Now we consider the case where the weights of the items are within a factor of $\delta$, that is, 
$w_1\le w_2\le\cdots \le w_n \le (1+\delta) w_1$, and $k$  elements are randomly sampled with replacements from these  $n$ elements. 
Let $Pr_{n, k, \delta}$ be the probability that all $k$ samples are distinct in this case. Then we have the following property.
\begin{lemma}\label{lemma:pr-n-k-delta} 
If $k < \tfrac{n}{2 (1 + \delta)}$, then
$e^{{-\tfrac{(1+\delta)k^2}{2n}}-\tfrac{2(1+\delta)^2 k^3 }{3 n^2}} \le  Pr_{n,k,\delta} \le e^{-\tfrac{(k-1)k}{2n(1+\delta)}} \enspace.$
\end{lemma} 
\begin{proof}
 We first consider the lower bound of $Pr_{n,k,\delta}$. By assumption,  $w_1\le w_2\le\cdots \le w_n \le (1+\delta) w_1$. Plugging this into the lower bound from   Lemma \ref{lemma:pr-n-k},  we get $$  Pr_{n, k, \delta}  \ge   \prod_{j = 1}^{k-1} \left( 1 - \tfrac {\sum_{i=1}^{j} w_{n-i+1}}{W} \right)   \ge 
  \prod_{j = 1}^{k-1} \left( 1 - \tfrac {  (1 + \delta)w_1 \cdot j} {nw_1}  \right)   =  \prod_{j = 1}^{k-1} \left( 1 - \tfrac {  (1 + \delta) \cdot j}{n} \right)  $$
  Since $  k < \tfrac{n}{2 (1 + \delta)}$, for any $j < k$, we have 
  $ \tfrac {  (1 + \delta) j}{n} < \tfrac{1}{2}$. Using the fact $(1-x) \ge e^{-x - 2x^2}$ for $ 0 \le x \le \tfrac{1}{2}$, we get the following 
 \begin{eqnarray*}Pr_{n, k, \delta}  & \ge &
   \prod_{j = 1}^{k-1} \left( 1 - \tfrac {  (1 + \delta) j}{n} \right)      \ge  
 \prod_{j = 1}^{k-1} e^{-\tfrac {   (1 + \delta) j}{n} -\tfrac { 2  (1 + \delta)^2 j^2 }{n^2}}  \\
 & > & e^{-\tfrac{(1+\delta)k(k-1)}{2n}-\tfrac{2(1+\delta)^2 k^3} {3 n^2}} 
   >  e^{-\tfrac{(1+\delta)k^2}{2n} - \tfrac{2(1+\delta)^2 k^3} {3 n^2}}
  \end{eqnarray*}

Now we  establish  the upper bound of $Pr_{n,k, \delta}$. 
Since, $w_i \le w_{i+1}$ for $1 \le i \le n-1$ and $w_n \le (1+\delta) w_1$, plugging into the upper bound of Lemma \ref{lemma:pr-n-k},   we get  
$$  Pr_{n,k, \delta}   \le   \prod_{j = 1}^{k-1} \left(1-\tfrac{ \sum_{i=1}^{j} {w_i}}{W} \right) \le \prod_{j = 1}^{k-1} \left(1-\tfrac{     {w_1} \cdot j}{n \cdot (1+\delta)w_1} \right) \le \prod_{j = 1}^{k-1} \left(1-\tfrac{  j }{n \cdot (1+\delta)} \right). $$ 
Using the fact 
$(1-x) \le e^{-x}$, we obtain $$  Pr_{n,k,\delta} \le \prod_{j = 1}^{k-1} e^{-\tfrac{  j }{n \cdot (1+\delta)}} = e^{-\sum_{j=1}^{k-1} \tfrac{  j }{n \cdot (1+\delta)}} = e^{- \tfrac{  (k-1)k }{2n \cdot (1+\delta)}}.$$ 

\end{proof}

For convenience, we let $Pr_{ub}( n, k, \delta)$  and $Pr_{lb}( n, k, \delta)$ denote the upper and lower bound of $Pr_{ n, k, \delta}$ from the above lemma, respectively. That is, 
\begin{equation*}\label{Pr_{n,k,delta}:upper bound}  Pr_{ub}( n,k,\delta ) = e^{-\tfrac{(k-1)k}{2n(1+\delta)}}
\end{equation*} 
and   
\begin{equation*}\label{Pr_{n,k,delta}:lower bound}Pr_{lb}( n,k,\delta ) =   e^{-\tfrac{(1+\delta)k^2} {2n} - \tfrac{2(1+\delta)^2 k^3 } {3 n^2}} \enspace.
\end{equation*}  

Next, we  consider these two bounds of $Pr_{n,k,\delta}$, $Pr_{ub}( n, k, \delta)$  and $Pr_{lb}( n, k, \delta)$, when the number of samples $k$ is within a constant factor of  $\sqrt{n}$. 

 \begin{lemma}\label{Pl-Pu-lemma} 
 Let ${c}$ be a constant such that  $ 4 \le {c} \le \tfrac{1}{4 \delta}$,   and $n \ge \tfrac{1}{\delta^2}$, then we have 
    
    \begin{equation}\label{P_L}
         Pr_{lb}( n, \tfrac{\sqrt{n}}{(1+\delta)^{ {c}}}, \delta) 
         \ge \tfrac{1}{ (1 - \delta) \sqrt {e}}   
         \enspace,
    \end{equation}

    \begin{equation}\label{P_R-upper-bound}  
Pr_{ub}( n,(1+\delta)^{{c}}  \sqrt{n}, \delta) 
\le \tfrac{1}{(1+ \delta) \sqrt {e}} \enspace,
\end{equation}

\begin{equation}\label{P_R-lower-bound}  
Pr_{ub}( n,(1+\delta)^{{c}}  \sqrt{n},\delta) \ge \tfrac{1}{ e } \enspace.
\end{equation}
\end{lemma} 
\begin{proof}  
We first prove the inequality (\ref{P_L}). Let the number of samples be  $k= \tfrac{\sqrt{n}}{(1+\delta)^{{c}}}$. Plugging this into the definition of $Pr_{lb}( n, k, \delta)$, we get

$$   Pr_{lb}( n, \tfrac{\sqrt{n}}{(1+\delta)^{{c}}}, \delta)  
    =   e^{-\tfrac{(1+\delta)^{1-2{c}} } {2} - \tfrac {2} { 3 \sqrt{n} (1+\delta)^{3{c} -2} }}  \enspace.$$
By assumption,  $ {c} \ge 4$  and    $n \ge \tfrac{1}{\delta^2}$,    $\sqrt{n} (1+\delta)^{3{c} -2} \ge \sqrt{n} \ge \tfrac{1}{\delta}$. Thus,
    
  \begin{eqnarray*}\label{P_L_0}  
   Pr_{lb}( n, \tfrac{\sqrt{n}}{(1+\delta)^{{c}}}, \delta)  
    &=&  e^{-\tfrac{(1+\delta)^{1-2{c}} } {2} - \tfrac {2} { 3 \sqrt{n} (1+\delta)^{3{c} -2} }}  \\
     & \ge &  e^{-\tfrac{(1+\delta)^{1-2{c}} } {2}} \cdot e^{- \tfrac {2\delta} { 3}}  \\
    &\ge&  e^{-\tfrac{1}{2}(1- (2{c}-1)\delta +(2{c}-1){c} \delta^2)} \cdot e^{- \tfrac {2\delta} { 3}} \\
    &\ge&  e^{-\tfrac{1}{2}(1-{(2{c}-1)\delta} + \tfrac{(2{c}-1)\delta}{4}) }    \cdot e^{- \tfrac {2\delta} { 3}}   \hspace{0.4in}  \text { by }  {c} \le \tfrac{1}{4\delta}\\
    &\ge&  e^{-\tfrac{1}{2}(1-\tfrac{3(2{c}-1)\delta}{4}) -\tfrac{2\delta}{3}} \\
    &\ge&  
    e^{-\tfrac{1}{2} + \tfrac{15 \delta }{8} }
    \hspace{1.0in}  \text { by } {c} \ge 4  \\ 
    &\ge& \tfrac{1}{\sqrt{e}} \cdot (1 + \tfrac{15 \delta}{8}) 
    \hspace{0.6in}  \text { by } e^x \ge 1+x \\
    &\ge& \tfrac{1}{\sqrt{e}(1-\delta)}
\end{eqnarray*}

\medskip

Next, we prove the inequality (\ref {P_R-upper-bound}). 
By definition of $Pr_{ub}( n, k, \delta)$, we have 
$Pr_{ub}(n, k, \delta) =  e^{-\tfrac{(k-1)k}{2n(1+\delta)}} \le  e^{-\tfrac{(k-1)^2}{2n(1+\delta)}} .$ Let  $k= (1+\delta)^{{c}}\sqrt{n}$, and ${c'}={c}-1$, we have

\begin{eqnarray*}\label{P_R_0} 
  Pr_{ub}( n,(1+\delta)^{{c}}\sqrt{n}, \delta)  
  &\le&  e^{-\tfrac{((1+\delta)^{{c}}\sqrt{n}-1)^2}{2n(1+\delta)}}   \\
    &\le&  e^{-\tfrac{((1+\delta)^{{c'}}\sqrt{n})^2}{2n(1+\delta)}}  \hspace{0.4in}  \text { by } {c} = {c'}+1  \text { and } n \ge \tfrac{1}{\delta^2}\\
    &=& e^{-\tfrac{(1+\delta)^{2{c'}-1}}{2}} \\
    &\le& e^{-\tfrac{1+(2{c'}-1)\delta}{2}} \hspace{0.4in}  \text { by } (1+\delta)^x \ge \delta x \text { for } \delta>0 \text { and } x \ge 0\\
    &\le& e^{-\tfrac{1}{2}}\cdot e^{-\tfrac{(2{c'}-1)\delta}{2}}\\
    &\le& e^{-\tfrac{1}{2}} \cdot (1-\tfrac{1}{2}\tfrac{(2{c'}-1)\delta}{2}) \hspace{0.1in}  \text { by } {c}\le \tfrac{1}{4\delta} \text{ and } e^{-x} 
    \le 1-\tfrac{x}{2} \text { if } x<1 \\
    &\le& \tfrac{1}{\sqrt{e}}(1-\tfrac{5\delta}{4}) \hspace{1in}  \text { by } {c}\ge 4 \text { and } {c'}={c}-1\\
    &\le& \tfrac{1}{\sqrt{e}(1+\delta)}
  \enspace.
 \end{eqnarray*}

On the other hand, we have 
$ Pr_{ub}(n, k, \delta) =   e^{-\tfrac{k(k-1)}{2n(1+\delta)}}  \ge 
 e^{-\tfrac{k^2}{2n(1+\delta)}} $.   Let  $k=(1+\delta)^{{c}}\sqrt{n}$, we get  
 \begin{eqnarray*} 
 Pr_{ub}( n,(1+\delta)^{{c}}\sqrt{n}, \delta) & \ge &  e^{-\tfrac {(1+\delta)^{2{c}-1}}{2}} \\
& \ge & e^{-\tfrac{1}{2}(1+\delta)^ {\tfrac{1}{2 \delta}}} 
  \hspace{0.2in}  \text { by }  {c} \le \tfrac{1}{4\delta} \\
& \ge  & e^{-\tfrac{1}{2}\sqrt {e}}
 \hspace{0.4in}  \text { by } 
 \ \   (1+\delta)^{\tfrac{1}{ \delta}} < e \\  
&\ge &   \tfrac {1}{e} 
\end{eqnarray*}

\end{proof}

We are ready to design and analyze sublinear time algorithms for generating input sketch of a given instance using weighted random sampling. From now on, whenever we draw samples from the $n$ jobs, we use weighted random sampling with replacement, and we differentiate between a sample and a job since a job may be sampled more than once.

Our algorithms utilize the generalized paradox to estimate the number of jobs in an interval. The samples are partitioned into many groups, if the number of groups with collisions among $k$ samples exceeds a certain threshold, we infer there are approximately $k^2$ jobs in the interval.
 
\subsection {Generating a Sketch When \texorpdfstring{$n$}~ is Known}\label{subsec: Sublinear1-sec}
In this section, we present  a sublinear time  algorithm, the Sketch-Algorithm,  to obtain an $(\alpha,\beta_1,\beta_2)$-sketch when the number of jobs $n$ is known.  

The algorithm has two main steps. In the first step, an upper bound of $p_{max}$, $p_{max}'$, is obtained by drawing a constant number of samples. Then,   we define  the processing time intervals $I_k = (\tilde{p}_{k+1}, \tilde{p}_k]$ where $\tilde{p}_k = p'_{max} \cdot  (1-\delta)^{k-1}$.

In the second step, the sketch instance $\tilde{I}$ is obtained by drawing $\tilde{O}(\sqrt{n})$ samples.  
We will only include the intervals $I_k$ with  $k \le h = O(\log n)$ and with sufficiently many samples, discarding all others.  
To estimate the number of jobs in each remaining intervals $I_k$, we distinguish two cases based on whether $I_k$ contains a job that is sampled $\Omega(\log n)$ times. Let $H_1$ and $H_2$ denote the sets of intervals in these two cases, respectively.
We use different methods to estimate the corresponding number of jobs in the intervals of $H_1$ and $H_2$. For each interval in $H_1$, with high probability, it contains jobs with the large processing time, and every job  in the interval from $I$ has been sampled many times. Thus, we can directly estimate the number of jobs by counting the number of distinct sampled jobs.   For each interval in $H_2$, we estimate the number of jobs using a generalized birthday paradox argument. 
The details of the algorithm are given below.  For convenience, we use  the following notations:
\begin{enumerate} 
\item[]  $x_j$:  the number of times job $j$ is sampled 

\item[] $X_k$:  the number of samples   whose processing times are in $I_k$.
\end{enumerate} 
 
{\noindent \bf Sketch-Algorithm}

{\noindent \bf Input:}   

\hspace{0.2in} the number of machines, $m$

\hspace{0.2in}  an instance $ I = \{ p_i: 1 \le i \le n\}$

\hspace{0.2in}  a fixed parameter $\gamma_0\in (0, \tfrac{1}{12}]$ to control the failure rate

\hspace{0.2in}  the approximation error $\delta\in (0, \tfrac{1}{2})$ 
  
{\noindent \bf Output:}   a sketch instance $ \tilde{I} = \{\langle \tilde {n}_i, \tilde {p}_i \rangle  : 1 \le i \le h\}$ where $ h = \tfrac{1}{\delta} \ln \tfrac{n^2}{\delta} $ 

\begin{enumerate} [1.]

\item \label{step1} Obtain an upper bound of $p_{max}$, $p'_{max}$

\begin{enumerate}[a.] 

\item\label{k0-line}
draw $K_0 =  \lceil \log \tfrac{1}{\gamma_0} \rceil $ samples 

\item \label{p'-line}
let $p'$ be the  maximum processing time among the $K_0$ samples 
\item \label{I_k-line} let 
$p'_{max} =2np'$

$\tilde{p}_k = p'_{max} \cdot  (1-\delta)^{k-1}$, and   $I_k = (\tilde{p}_{k+1}, \tilde{p}_k]$  for $k \ge 1$

\end{enumerate}

\item \label{step2} Obtain the sketch instance $\tilde{I}$

\begin{enumerate}[a.] 
    \item\label{define-h} let $h = \tfrac{1}{\delta} \ln ( \tfrac{n^2}{\delta})$

\item\label{define-tau} 
 let $\tau (n)  = \delta h \ln(16 h (\log_{1+\delta} 3 \sqrt {n} )/\gamma_0)$

\item\label{define-k}
draw $K=\tfrac{36 m}{\delta^4} \sqrt{n} \cdot \tau(n)$ samples 

\item \label{step:intervals-filter} let  $H = \{ I_k: k \le h, \text{ and }  X_k \ge \tfrac{\delta}{mh} K \}$

 \item\label{f(.)-def-line} let  $f(n) = 36 \ln (\tfrac{2n}{ \gamma_0})$

\item   let $H_1 = \{ I_k: I_k \in H, \text { and } \exists j, \text{ such that } j \in I_k, \text { and } 
x_j \ge f(n)  \}$ and $H_2 = H \setminus H_1$ 

\item  let  $\tilde{I_1} = \tilde{I_2} = \emptyset$  
\item for each  interval $I_k \in H_1 $  
 
\item []
\hspace{0.1in} let $\tilde{n}_k$ 
be the number of  jobs (not samples) from $I_k$ that are sampled 
\item[]
\hspace{0.1in}  let $\tilde{I_1} = \tilde{I_1}  \cup \{ \langle \tilde{n}_k, \tilde{p}_k\rangle\}$

\item \label{step:n_k}  
  
for each  interval $I_k \in H_2$
 
\begin{enumerate}[i.]
 
\item let $h_0= 3\sqrt{n}$ 
\item $\tilde{n}_k=1$
\item partition the $X_k$ samples from $I_k$ 
into groups of size $h_0$, $G_1,\cdots, G_u$ 

\item \label{step:finding-l} for  $l=1$, $l\le h_0$, $l=l(1+\delta)$

\item[] \qquad  let $g_{l,k}$ be the number  of groups where the first $l$ samples

\qquad in the group are all distinct 

\item[] 
\qquad 
if $g_{l,k}\le \tfrac{1}{\sqrt {e} } u$
\item[] 
\hspace{0.5in}   let $\tilde{n}_k=l^2$, go to step~\ref{step:n-k-estimate-found}  
\item \label{step:n-k-estimate-found} $\tilde{I_2} = \tilde{I_2}  \cup \{ \langle \tilde{n}_k, \tilde{p}_k\rangle\}$
\end{enumerate}
 
\item  $ \tilde{I} = \tilde{I_1} \cup \tilde{I_2} $
\end{enumerate}

\end{enumerate}
\medskip
Now, we analyze the sketch instance $\tilde{I}$ computed by the algorithm. We first show that    
$p'_{max}$ is a good upper bound for $p_{max}$.

\begin{lemma}\label{lemma:pmax}
$Pr(p_{max} < p'_{max}) \ge 1 - \gamma_0$.
\end{lemma}
\begin{proof} Let $P= \sum_{j=1}^n p_j$. Without loss of generality, we assume that $p_1 \le \ldots \le p_n$, let $mid$ be the largest index such that   
$\sum_{ 1\le j \le  mid}p_j < \tfrac{P}{ 2}$. Then the probability that all $K_0 $ samples  have processing time at most $p_{mid}$ is at most $(\tfrac{1}{2})^{K_0} = (\tfrac{1}{2})^{\lceil \log \tfrac{1}{\gamma_0} \rceil} \le \gamma_0 $. 
That implies, with probability at least $1-\gamma_0$,  the largest processing time among the samples, $p'$, is greater than $p_{mid}$. 
By the choice of $p_{mid}$, we must have $\sum_{p_j\le p'}p_j \ge \tfrac{P}{ 2}$, which implies that $ P \le 2 \sum_{p_j\le p'}p_j  \le 2n p' = p'_{max}$.
Consequently,  $p_{max}  < P \le    p'_{max} $. 
\end{proof}

In the following, we show that the Sketch-Algorithm computes a good estimation for $n_k$, $1 \le k \le h$. 

\begin{lemma}\label{Approx-lemma-n_k_1}
The probability that there exists an interval $I_k$ such that  $\langle \tilde{n}_k, \tilde{p}_k\rangle \in \tilde {I_1} $ and   $\tilde{n}_k \neq  n_k$ is at most  $\gamma_0$.
\end{lemma} 
 \begin{proof}
 We first use Chernoff bound to show that a small job is unlikely to be sampled many times (i.e. $x_j \ge f(n)$) where a job $j$ is small if $p_j < \tfrac{2 f(n)}{3 K} P$, otherwise we say job $j$ is big. 
  Applying Chernoff bound with $p =\tfrac{2 f(n)}{3 K} $  
for a single small job $j$, we get 

\begin{equation} \label{small-job-unlikly-sampled-many-times}
\begin{split}
 Pr( x_j \ge f(n)) & =   Pr( x_j \ge (1+ \tfrac{1}{2} ) \tfrac{2 f(n)}{3 K} K ) \\   &  <    e^{-{\tfrac{1}{3}} \left(\tfrac{1}{2} \right) ^2 \cdot \tfrac{2 f(n)}{3 K} K }  =  e^{-{  \tfrac{f(n)}{18}}   } = e^{ - \tfrac{36 \ln (\tfrac{2n}{ \gamma_0})}{18} }
 \le  (\tfrac{\gamma_0}{2 n} )^2
<  \tfrac{\gamma_0}{2 n} \enspace.
\end{split}
\end{equation}

Next, we show that if $j$, $j  \in I_k$, is a big job, i.e., $p_j \ge \tfrac{2 f(n)}{3 K} P$,  then  every other job $r$ in $I_k$ is sampled at least $\tfrac{f(n)}{6}$ times with high probability.   
By definition of intervals, we know that  $p_r \ge (1 - \delta) p_j$. Since  $\delta < \tfrac{1}{2}$,  if $p_j \ge \tfrac{2 f(n)}{3 K} P$, then $p_r\ge (1 - \delta) p_j \ge  \tfrac{f(n)}{3 K} P$. Thus, the probability of job $r$ being sampled is at least $\tfrac{ f(n)}{3 K} $.  When we draw $K$ samples, by applying Chernoff bound, we can get 
\begin{equation}\label{all-jobs-sampled-many-times}
\begin{split}
Pr(x_r < \tfrac{1}{6} f(n) ) & =   Pr(x_r < (1-\tfrac{1}{2}) \tfrac{f(n)}{3 K} K ) \\ 
& \le  
e^{-{\tfrac{1}{2}} \left(\tfrac{1}{2} \right) ^2 \cdot \tfrac{ f(n)}{3 K} K }  = 
e^{-{  \tfrac{f(n)}{24}}} =    e^{  -\tfrac{36 \ln (\tfrac{2n}{ \gamma_0})}{24} } \le \tfrac{\gamma_0}{2n} \enspace. 
\end{split}
\end{equation}

By definition, every  $I_k \in H_1$ contains a job $j$ that has been sampled at least $f(n)$ times. Applying the union bound to all these jobs and using (\ref{small-job-unlikly-sampled-many-times}), the probability that one of these jobs  is   small  is  at most $ \tfrac{\gamma_0}{2}$.  
 Again, applying the union bound and using  (\ref{all-jobs-sampled-many-times}), the probability that there exist a job $r$ that is in the same interval as some of these large jobs $j$ 
but   not sampled, is at most $ \tfrac{\gamma_0}{2}$.  
In summary, the probability that there exists an interval $I_k$ such that  $\langle \tilde{n}_k, \tilde{p}_k\rangle \in \tilde {I_1} $ and   $\tilde{n}_k \neq  n_k$ is at most
$ \gamma_0$.
\end{proof}

\begin{lemma}\label{Approx-lemma-n_k_2}
The probability that there exists an interval $I_k$ such that  $\langle \tilde{n}_k, \tilde{p}_k\rangle \in \tilde {I_2} $ and   $\tilde{n}_k \notin  [(1 + \delta)^{-2 {c}} \cdot n_k,  (1 + \delta)^{2 {c}+2} \cdot n_k])$, ${c} = 4 $,  is at most $\tfrac{\gamma_0}{8}$.
\end{lemma} 

  \begin{proof}
For any $I_k$ such that  $\langle \tilde{n}_k, \tilde{p}_k\rangle \in \tilde {I_2} $,   let $l_k = \sqrt { \tilde{n}_k} $. To prove the lemma, it is equivalent  to consider the probability of  
$ l_k \notin [ \tfrac{\sqrt{n_k}}{{{(1+\delta)^{{c}}}}} ,(1+\delta)^{{c}+1}\sqrt{n_k}]$.  
We will show that for a particular interval $I_k$, 
\begin{equation}\label{l_k}
     Pr(l_k \notin [ \tfrac{\sqrt{n_k}}{{{(1+\delta)^{{c}}}}} ,(1+\delta)^{{c}+1}\sqrt{n_k}]) \le  \tfrac{\gamma_0}{8h}
     \enspace.
\end{equation} 
Then, 
 applying union bound to all (at most $h$) intervals $I_k$ for which $\langle \tilde{n}_k, \tilde{p}_k\rangle \in \tilde {I_2} $, the probability that there exists one interval $I_k$ that $l_k \notin [ \tfrac{\sqrt{n_k}}{{{(1+\delta)^{{c}}}}} ,(1+\delta)^{{c}+1}\sqrt{n_k}]$  is at most $\tfrac{\gamma_0}{8}$. Hence,   the lemma follows.

We prove (\ref{l_k}) in two cases, $l_k  <\tfrac{\sqrt{n_k}}{{(1+\delta)^{{c}}}} $ and  $l_k > {(1+\delta)^{{c}+1}}{\sqrt{n_k}} $. 
First,   we consider the probability of $l_k  <\tfrac{\sqrt{n_k}}{{(1+\delta)^{{c}}}} $. From Step~
\ref{step:finding-l} of the algorithm, we know that if $g_{l,k} < \tfrac{1}{\sqrt{e}} u$ in an iteration $1 \le l < \tfrac{\sqrt{n_k}}{{(1+\delta)^{{c}}}}$,  then we have $\tilde{n}_k=l^2$ and $l_k = l <\tfrac{\sqrt{n_k}}{{(1+\delta)^{{c}}}} $. 
Thus, $Pr( l_k < \tfrac{\sqrt{n_k}}{{(1+\delta)^{{c}}}} ) \le \sum_{l} Pr (g_{l,k} < \tfrac{1}{\sqrt{e}} u)$ where $ l \le \tfrac{\sqrt{n_k}}{{(1+\delta)^{{c}}}} $.  
Since $l$ grows geometrically and $l \le 3 \sqrt {n}$, 
the number of such iterations is at most $  \log_  {1+\delta} 3 \sqrt{n}$. 
Thus,  
\begin{equation} \label{l-k-left} Pr( l_k < \tfrac{\sqrt{n_k}}{(1+\delta)^{{c}}} )  \le    (\log_  {1+\delta} 3 \sqrt{n} )  Pr (g_{l,k} \le \tfrac{1}{\sqrt{e}} u) 
\enspace. \end{equation}
Now, let us   bound $Pr (g_{l,k} \le \tfrac{1}{\sqrt{e}} u) $ in a particular iteration $l$. 
 Recall that $g_{l,k}$ is the number of sample groups in which the first $l$ samples are all distinct.  
By inequality (\ref{P_L}), the probability that the first $l$ elements in a particular group are all distinct is greater than  \begin{equation} \label{bound-of-pl} p_l = Pr_{lb}(n_k, \tfrac{\sqrt{n_k}}{{{(1+\delta)^{{c}}}}}, \delta)  \ge \tfrac {1} { (1-\delta)\sqrt{e} } 
\enspace. \end{equation}
Thus, we have $\tfrac{1}{\sqrt{e}} u \le (1 - \delta) p_l u$.  In addition, 
since the number of samples in $I_k$ is at least (see Step~(\ref{step:intervals-filter}) of Sketch-Algorithm) 
$$\tfrac{  \delta K } {m h} = 
  \tfrac{ \delta} {m  h } \tfrac{36 m}{\delta^4} \sqrt{n} \cdot \tau(n) = \tfrac{ 36} {  \delta^3 h}   \sqrt{n} \cdot \tau(n) , $$ 
and each group has  $3 \sqrt{n}$ samples,  the number of  groups $u$ satisfies the following inequality 
\begin{equation} \label{bound-of-u}
u    \ge \tfrac{ 12} {  \delta^3 h }    \cdot \tau(n)  
\ge \tfrac{ 12} {  \delta^2  }    \cdot 
   \ln(16 h (\log_{1+\delta} 3 \sqrt {n})/\gamma_0 )\enspace. \end{equation}
Applying the Chernoff bound to the random variable $g_{l,k}$ with $p_l$ and $\delta$, we have
\begin{eqnarray*}
Pr( g_{l,k} \le \tfrac{1}{\sqrt{e}} u ) &\le&  Pr (g_{l,k} < (1 -\delta) p_l  u ) \\ 
&\le&      e^{-\tfrac{\delta^2}{2} p_l u } \\
&\le&      e ^ {- \tfrac{\delta^2} {2\sqrt{e}} u} \hspace{1.2in} \text{ by (\ref{bound-of-pl})}\\
  &\le&   e ^ { - \ln (16h (\log_{1+\delta} 3 \sqrt{n})/\gamma_0)} \hspace{0.3in} \text{ by (\ref{bound-of-u})}\\
    &\le& \tfrac{\gamma_0}{16h (\log_{1+\delta} 3 \sqrt{n})}\enspace. 
\end{eqnarray*} 

Plugging into (\ref{l-k-left}),   $$Pr( l_k < \tfrac{\sqrt{n_k}}{(1+\delta)^{{c}}} )  \le    (\log_  {1+\delta} 3 \sqrt{n} )  Pr (g_{l,k} \le \tfrac{1}{\sqrt{e}} u)  \le \tfrac{\gamma_0}{16h} 
\enspace. $$

\medspace

Now, we consider the probability of $l_k > {(1+\delta)^{{c}+1}}{\sqrt{n_k}} $. 
From Step
\ref{step:finding-l} of the algorithm, we know this happens only if   $g_{l,k} \ge \tfrac{1}{\sqrt{e}} u$ for all iterations $l \le {(1+\delta)^{{c}+1}}{\sqrt{n_k}}$. 
In particular,  this is true for the $l$ where $ (1+\delta)^{{c}}\sqrt{n_k} < l \le (1+\delta)^{{c} + 1}\sqrt{n_k}$, and this implies that $$Pr(l_k > {(1+\delta)^{{c}+1}}{\sqrt{n_k}}) \le Pr( g_{l,k} \ge \tfrac{1}{\sqrt{e}}u ) \enspace.$$
Additionally,  since  $ l > (1+\delta)^{{c}}\sqrt{n_k} $,   we have $$Pr_{ub}(n_k, l, \delta) < Pr_{ub}(n_k, (1+\delta)^{{c}}\sqrt{n_k}, \delta)\enspace.$$ 
Let $p_r = Pr_{ub}(n_k, (1+\delta)^{{c}}\sqrt{n_k}, \delta$), by the inequalities 
 (\ref{P_R-upper-bound}) and (\ref{P_R-lower-bound}),
$ \tfrac{1}{e} <   p_r \le \tfrac{1}{(1+\delta)\sqrt{e}}  $. Therefore, 
$\tfrac{1}{\sqrt{e}}u \ge (1+ \delta) p_r u$, and we get 
\begin{eqnarray*}
Pr(l_k > {(1+\delta)^{{c}+1}}{\sqrt{n_k}}) &\le& Pr( g_{l,k} \ge \tfrac{1}{\sqrt{e}}u ) \\
 &\le& Pr (g_{l,k} \ge (1+\delta)p_r u)  \\ 
 &\le&    e^{-\tfrac{\delta^2}{3} p_r  u}  \hspace{0.3in} \text{by Chernoff bound  for $g_{l,k}$} \\
 &\le&    e^{-\tfrac{\delta^2}{9}   u }   \hspace{0.4in} 
 (p_r \ge \tfrac{1}{e} > \tfrac{1}{3}) \\
&\le &  e^{-\tfrac{\delta^2}{9}  \cdot  \tfrac{ 12} {  \delta^2  }    \cdot 
   \ln(16 h (\log_{1+\delta} 3 \sqrt {n})/\gamma_0 ) }
\le   \tfrac{\gamma_0}{ 16  h } \hspace{0.2in}\text{by inequality (\ref{bound-of-u})}
\enspace.
\end{eqnarray*}

\medskip

Combining the above two cases, for a particular interval $I_k$, we get
 $$Pr( l_k \notin [ \tfrac{\sqrt{n_k}}{{{(1+\delta)^{{c}}}}} ,(1+\delta)^{{c}+1}\sqrt{n_k}]) \le  Pr( l_k < \tfrac{\sqrt{n_k}}{(1+\delta)^{{c}}} ) + Pr( l_k >  {\sqrt{n_k}}{(1+\delta)^{{c}+1}} )  \le    \tfrac{  \gamma_0}
      { 8  h} \enspace.$$   
Applying union bound to all $h$ intervals,  we have 
$$Pr( l_k \notin [ \tfrac{\sqrt{n_k}}{{{(1+\delta)^{{c}}}}} ,(1+\delta)^{{c}+1}\sqrt{n_k}]) \le    \tfrac{\gamma_0}{8}.$$
\end{proof}

From Lemma~\ref{Approx-lemma-n_k_1}, and \ref{Approx-lemma-n_k_2}, we can get the following.

\begin{theorem}\label{Approx-lemma-n_k}
For all intervals $I_k$ where    $\langle \tilde{n}_k, \tilde{p}_k\rangle \in \tilde {I} $, the probability that  $\tilde{n}_k \in  [(1 + \delta)^{-2 {c}} \cdot n_k,  (1 + \delta)^{2 {c}+2} \cdot n_k])$ where ${c} = 4 $, 
is at least   $1-\tfrac{9}{8}\gamma_0$.
\end{theorem}

 Theorem  \ref{Approx-lemma-n_k} tells us that given an instance $I$, $\tilde{I}$ provides a good approximation for the number of jobs in the intervals of $H$. In the following, we prove that  $\tilde{I}$ also provides   a good approximation for  the makespan of the original instance $I$. During the construction  of $\tilde{I}$ from  $I$, 
 the Sketch-Algorithm introduces error in  three ways:
(1) we discard jobs from intervals $I_k$, $k > h$; 
(2) we discard jobs from intervals $I_k$ where $k \le h$ and $X_k < \tfrac{\delta}{mh} \cdot K$; 
and (3) for all the intervals $I_k \in H$, the processing times of the jobs are  rounded up by a factor of at most   $\tfrac{1} {1-\delta} \le   ({1 + \delta})$, 
We will show that the error introduced in each case is small.  

\begin{lemma}\label{lemma:small-job}
The total processing time of jobs in the interval 
$I_k$ with $ k > h$, 
is at most $ 2 \delta OPT$.    
\end{lemma}

\begin{proof} If $j$ is a job in $I_k$, $ k > h $, then its processing time $p_j$ is at most $ p'_{max} (1 - \delta)^{h} \le 2 n p' (1 - \delta) ^ h  $ where 
  $p'$ is the  maximum processing time among the $K_0$ samples (see Step~\ref{k0-line}). There are at most $n$ such jobs, and their  total processing time  is at most 
$$n \cdot 2 n   (1 - \delta) ^{h} p' \le 2 n^2 \cdot e ^{ - \delta h} p' \le 2 n^2 \cdot e ^{ - \delta h} OPT, $$
which is at most $ 2 \delta OPT $ since $h = \tfrac{1}{\delta} \ln \tfrac{n^2}{\delta} $
\end{proof}

\begin{lemma}\label{lemma:small-interval}
With probability at least $1-\gamma_0$, the total processing time of jobs from the intervals $I_k$,   $1 \le k \le h$ and 
$X_k < \tfrac{ \delta}{mh} \cdot K$, 
is at most $2\delta \text{OPT}$.
\end{lemma}
 \begin{proof}
Let  $I_k$ be an interval with $1 \le k \le h$ and $X_k   < \tfrac{ \delta}{mh} \cdot K$. 
Let $P(I_k)$ be the total processing time of jobs in $I_k$.   We prove that $P(I_k) \le  \tfrac{2\delta}{mh} P $ with high probability. 
Suppose that $P(I_k) \ge \tfrac{2\delta}{mh} P$.  Then, the probability  that  a  sample is from $I_k$  is at least   $\tfrac{ 2 \delta}{m h}$. When there are $K$ samples,  applying Chernoff bound, 
we get   $Pr (X_k < \tfrac{ \delta}{mh} \cdot K) =Pr (X_k < (1-\tfrac{1}{2})\cdot \tfrac{2 \delta}{mh} \cdot K)  \le  e^{ -\tfrac{1}{2} \left(\tfrac{1}{2}\right)^2 \tfrac{ 2 \delta}  {m h} K } = e^{- \tfrac{  {   \delta} } {4m h} K }$. 

Using union bound on all $h$ intervals, we get   the probability that there exists an $I_k$ 
  such that     $P(I_k) \ge  \tfrac{2\delta}{mh} P $ and $ X_k < \tfrac{ \delta}{mh} \cdot K$ is at most
$$ h \cdot e^{-\tfrac{  \delta}  {4 m h} K} \le    h \cdot e^{-( \tfrac{  \delta}  {4 m h} \cdot \tfrac{36 m}{\delta^4} \sqrt{n} \cdot \tau(n))}\\
   \le    h \cdot e^{-(\tfrac{9}{ h \delta^3} \sqrt{n} \cdot \tau(n))} \le    h \cdot \tfrac{\gamma_0}{ h} 
  \le  \gamma_0 
  \enspace.$$

  Therefore,  with probability at least $1-\gamma_0$, $P(I_k) <  \tfrac{2\delta}{mh} P $ for all $I_k$ where  $ X_k < \tfrac{ \delta}{mh} \cdot K$.
 %
  %
%
 The total processing time of  the intervals where $X_k <   \tfrac { \delta}{m h} K $   is at most $$  h \cdot \tfrac {2  \delta}{m h} P = 2 \delta \tfrac{P} {m }  \le  2\delta OPT  \enspace.$$   
\end{proof}

\begin{lemma}\label{lemma:compress-error}
Let $I'$ be the instance corresponding to $\tilde{I}$ such that for any interval pair $ \langle \tilde{n}_k, \tilde{p}_k\rangle   \in \tilde{I}$, there is an interval pair $ \langle n_k, \tilde{p}_k\rangle \in I'$. Then $ (1 - 2 {c} m \delta ) OPT (I') \le OPT(\tilde{I}) \le (1 + 6 {c} m \delta ) OPT (I')$ for $\delta \le \tfrac{1}{4 {c}}$   and ${c} = 4$.
 \end{lemma} 

   \begin{proof} 
%
 %
 The difference of $I'$ and $\tilde{I}$ lies in the number of jobs. 
 By Theorem~\ref{Approx-lemma-n_k}, $  n_k (1 + \delta)^{-2 {c}} \le  \tilde{n}_k \le 
 (1 + \delta)^{2 {c}+ 2} n_k  $. 

 Let $S'$ and $\tilde{S}$ be an optimal schedule for $I'$ and $\tilde{I}$, respectively. Let $P'$ be the total processing time of all jobs in $I'$, that is, $P' = \sum_{k} n_k \tilde{p}_k$. Comparing $S'$ and $\tilde{S}$, for each interval $I_k$, we have two cases. In the first case,   $n_k  \le  \tilde{n}_k$.
 Using the fact $\delta \le \tfrac{1}{4 {c}}$   and ${c} = 4$, 
 we have $$ (\tilde{n}_k - n_k ) \le [(1 + \delta)^{2 {c} +2 } -1 ] n_k \le     (2 {c} +2) \delta \cdot \sum_{i \ge 1 } \tfrac{1}{i!}  \cdot n_k    \le 3 c \delta \cdot (e-1) \cdot n_k \le  6 {c} \delta n_k.$$ In the second case, $\tilde{n}_k \le n_k $. We have $$ ( n_k - \tilde{n}_k ) \le [ 1 - (1 + \delta)^{- 2 {c}}] n_k  \le  [ 1 - (1 - \delta)^{2 {c}}] n_k \le    1 - (1 - 2{c} \delta) n_k \le 2 {c} \delta  n_k.$$  
Combining all the intervals, 
the makespan of $\tilde{S}$ is  greater than that of $S'$ by no more than $6 {c} \delta P'$ and less than that of  $S'$  by at most  $2 {c} \delta P'$. Since $OPT (I') \ge P'/m $, we get 
$ (1 - 2 {c} m \delta ) OPT (I') \le OPT(\tilde{I}) \le (1 + 6 {c} m \delta ) OPT (I')$.
 \end{proof}

Note that from $I$ to $I'$,  the processing times of jobs   are rounded up by a factor of at most $(1+\delta)$. This will increase the optimal makespan by a factor of at most $\delta$. In addition, we discard jobs from intervals not in $H$ whose total processing time is at most $4 \delta OPT (I)$ by  Lemmas~\ref {lemma:small-job} and  \ref{lemma:small-interval}. Thus, we have 
$$ (1 - 4 \delta) OPT(I) < OPT(I') < ( 1+ \delta)OPT(I) \enspace.$$
Together with Lemma~\ref{lemma:compress-error},  we get the following theorem.
\begin{theorem}\label{theorem: original-vs-compressed}
With probability at least $1-\tfrac{17}{8}\gamma_0$,   
$$(1 - 2 {c} m \delta) (1 - 4 \delta) OPT(I)  \le OPT(\tilde{I})  \le (1 + 6 {c} m\delta) (1 + \delta)  OPT(I), $$ where ${c} = 4$ and $\delta \le \tfrac{1}{4 c }$. 
\end{theorem}

To summarize, we have shown  that, with high probability, the Sketch-Algorithm computes an $(\alpha, \beta_1, \beta_2)$ sketch where  
$\alpha = (1 + \delta)^{2c+2}$ (Theorem \ref{Approx-lemma-n_k}),   $\beta_1 = (1 + 6 {c} m\delta) (1 + \delta)$ (Theorem \ref{theorem: original-vs-compressed}), $\beta_2 = 4 \delta$ (Lemma  \ref{lemma:small-job} and  \ref{lemma:small-interval}). The size of the sketch is $O(h) = O(\tfrac{1}{\delta} \ln {\tfrac{n^2}{\delta}}) $. The time complexity of Sketch-Algorithm is dominated by Step~\ref{step:n_k}, which computes $\tilde{n}_k$ for each interval $I_k \in H_2$. 
For each $I_k$, this can be done in $O(X_k)$ if we use a hash table for detecting identical samples in each group of $l \le h_0$ samples, where $l=1, (1+\delta), (1+\delta)^2, \cdots$. In total, the algorithm takes   $O(K) = O( \tfrac{36 m}{\delta^4} \sqrt{n} \cdot \tau(n))  = \tilde{O}(\tfrac{m}{\delta^4} \sqrt{n})$. 
Together with the approximation scheme from Section   \ref{subsec:approx-for-sketch},  we are ready to prove the Theorem~\ref{theorem:sublinear-n-known}.

\medskip

\begin{repeatedtheorem}{\ref{theorem:sublinear-n-known}}    Let $A(n, 1+\epsilon)$ be the time complexity for $(1+\epsilon)$-approximation scheme for makespan scheduling of $n$ input jobs. 
When the number of jobs $n$ is known, 
using non-adaptive weighted sampling, there is a randomized $(1+\epsilon)$-approximation algorithm for the   makespan minimization problem  that runs in $\tilde{O}(\tfrac{m^5}{\epsilon^4} \sqrt{n}  +  A(\ceiling{\tfrac{m} {\epsilon}}, m, {\epsilon} ) )$ time. Furthermore, it can compute a sketch schedule $\tilde{S}$, represented using   
$O(\tfrac{m^2}{ \epsilon^2} (\log \tfrac{n m}{\epsilon}))$ space,
which can be used to generate a schedule  of the jobs subsequently  with a  makespan of at most $(1 + 3 \epsilon)$ times the optimal.
 \end{repeatedtheorem}
 
 \medskip
 \begin{proof}
Given instance $I$, and parameters $\epsilon$ and $\gamma_0$, we follow the Main Algorithm. 
We first  invoke the Sketch-Algorithm to get the sketch instance $\tilde{I}$ and then apply the Meta-algorithm  to  $\tilde{I}$   with appropriate parameters. For example, let  $\delta = \tfrac{\epsilon} {12 {c} m}$, and $\alpha  = \tfrac{\epsilon}{4}$,   
by Theorems~\ref{theorem: original-vs-compressed} and Theorem~\ref{trans-thm}, we can get a $(1+\epsilon)$-approximation with probability at least $(1 - 3 \gamma_0)$.    The time complexity is $\tilde{O}(\tfrac{m^5}{\epsilon^4} \sqrt{n})$ for the Sketch-Algorithm,   and $A(\lceil\tfrac{m}{ \epsilon}\rceil, m, {\epsilon} )$ for the approximation scheme. Furthermore, by Theorem~\ref{meta-sketch}, the Modified-Meta-algorithm can generate a sketch schedule of size $O(mh)=O(\tfrac{m^2}{ \epsilon^2} (\log \tfrac{n m}{\epsilon}))$, which can be used to generate a real schedule whose makespan is at most $(1+ 3 \epsilon) OPT(I)$, as follows from Theorem\ref{theorem:sketch-to-schedule}.
\end{proof}

\subsection{Generating a Sketch When  \texorpdfstring{$n$}~ is Unknown}\label{subsec: Sublinear2-sec}
In this section, we develop a sublinear time algorithm, Adaptive-Sketch-Algorithm, for generating input sketch $\tilde{I}$ when  the   number of jobs, $n$, is unknown. 
Similar to the Sketch Algorithm in Section~\ref{subsec: Sublinear1-sec}, the Adaptive-Sketch-Algorithm has two main steps. First, we determine the intervals $I_k$ to be included in the sketch instance $\tilde{I}$; then   we estimate the number of the jobs in each of the intervals. 

In the first step, because $n$ is unknown, there is a slight difference in determining the intervals $I_k$. Instead of defining the intervals based on an upper bound of $p_{max}$, we use a large job processing time  $w_0$, where the jobs with processing times   greater than $w_0$ have a very small overall impact on the makespan. 

In the second step, we first determine $h$, the upper bound on the number of intervals  $\tilde{I}$. Again, there is a slight difference since $n$ is unknown. we choose $h$ in such a way that jobs falling outside these intervals contribute minimally to the makespan.
 We then draw the first round of samples and try to construct $\tilde{I}$. As before, we consider only the intervals $I_k$ for $ k \le h$ that have sufficiently many samples. 
 These intervals are divided into two sets, $H_1$ and $H_2$. While we can use the same method to estimate the number of jobs in the intervals of $H_1$, there is a key difference in estimating the number of jobs in the intervals of $H_2$ due to the unknown value of $n$. For these intervals, we may not be able to obtain an estimate using the birthday paradox from the first round of sampling, as we did in the Sketch Algorithm. 
Therefore, we keep increasing the sample size geometrically and resample again.  In each round, we try to estimate the number of jobs for as many intervals as possible. This process is repeated until an estimate can be obtained for all intervals.  

In the following we formally present the algorithm. We will continue to use $X_k$ to denote  the  number of samples from the interval $I_k$,
and $x_j$ to denote  the  number of times that job $j$ occurs in the  samples. When the   algorithm draws a new collection of samples, the $X_k$ and $x_j$ are updated based on the new  samples.  

\medskip

{\noindent \bf Adaptive-Sketch-Algorithm}

Input: the instance $ I = \{ p_i: 1 \le i \le n\}$   

\hspace{0.33in}   the number of machines $m$ 

\hspace{0.33in}   the approximation error $\delta\in (0,\tfrac{1}{3}]$, and

\hspace{0.33in}    the  failure rate parameter  $\gamma_0\in (0, \tfrac{1}{12}]$.


Output: a sketch instance $ \tilde{I} = \{\langle \tilde {n}_i, \tilde {p}_i \rangle\}$\\
 
Steps: 
\begin{enumerate}[1.] 
  
\item \label{adaptive-step1}  determine the intervals $I_k$
\begin{enumerate}[a.] 
     \item\label{line-d0} let  $d_0$ be the smallest integer such that
$(1-\tfrac{\delta}{m})^{\tfrac{m}{\delta}d_0 }\le \gamma_0$

    \item draw ${K}=\tfrac{50 m}{  \delta} d_0$ random samples,
    and let  $w_0$ be the largest processing time  among the samples.

   \item\label{define-Ik-line}  for all $k \ge 1$, let 
 $I_k = (\tilde{p}_{k+1}, \tilde{p}_k]$   where $\tilde{p}_k = w_0 \cdot  (1-\delta)^{k-1}$ 

\end{enumerate}

\item \label{adaptive-step2} obtain the sketch instance $\tilde{I}$
\begin{enumerate}[a.] 

    \item\label{h-assignment-line}  let $h$ be the smallest integer  such that $\sum_{k > h} X_k \le \tfrac{\delta {K}}{ 4m} = \tfrac{d_0}{4}$

\item\label{adaptive-k0-line} let  $K_0= \tfrac{8}{ \beta_0}\cdot \ln \tfrac{e}{\beta_0 \gamma_0}$ where $\beta_0=\tfrac{\delta^3}{32mh}$ 

\item draw $K_0$ random samples

\item\label{select-H-line}   let $H$ be the set of intervals $H = \{ I_k: 1 \le  k \le h, \text {and }  X_k \ge \tfrac {\delta}{m h} \cdot {K_0}\}$  

\item\label{select-H1-line}   let $H_1 = \{ I_k: I_k \in H, \text { and } \exists j$, $j \in I_k$ and  $ x_j \ge  8 \beta_0 K_0 \}$, and $H_2 = H \setminus H_1$

\item  let  $\tilde{I_1} = \tilde{I_2}  = \emptyset$ 
 
\item 
for each  interval $I_k \in H_1 $  
 
\hspace{0.2in} let $\tilde{n}_k$ be the number of sampled jobs from $I_k$

\hspace{0.2in}  $\tilde{I_1} = \tilde{I_1}  \cup \{ \langle \tilde{n}_k, \tilde{p}_k\rangle\}$

\item 
for each interval $I_k \in H_2 $ let $gs_k=1$
\item let $j = 1$  

\item \label{step:n_k_adapted}  while there exists $I_k \in H_2$ that are not marked
\begin{enumerate}[i.]
\item \label{Kj-line}  discard the current samples, draw $K_j = 2^j K_0$ new samples

\item \label{step:estimate} for each unmarked interval $I_k$
\begin{enumerate}[A.]
    \item find the largest integer $t$ such that  $X_k \ge l_t \cdot  u_t$ where $l_t = (1+ \delta)^t$, $u_t= \lceil \tfrac{ 3 {e} }{\delta^2} (\ln\tfrac{1}{\gamma_t}) \rceil$ and $\gamma_t= \tfrac{\delta}{h l_t} \gamma_0$
    
    \item partition the samples in $I_k$ evenly  into $u_t$ groups $G_1,\cdots, G_{u_t}$     

    \item \label{h2-nk-single-round} for $i=gs_k$, $i \le t$ and $I_k$ still unmarked, $i=i+1$

    \item[] \qquad     let  $l_i = (1+ \delta)^i$, $\gamma_i= \tfrac{\delta}{h l_i} \gamma_0$, and $u_i= \lceil \tfrac{ 3 {e} }{\delta^2} (\ln\tfrac{1}{\gamma_i}) \rceil$     
    \item [] \qquad let  $g_{l_i, k}$ be the number of groups  among $G_1,\cdots, G_{u_i}$ 
    
          \qquad where the first $l_i$ samples in the group are all distinct   
    \item[]  \qquad if $ g_{l_i, k} \le \tfrac{1}{\sqrt {e} } u_i$ 
          
                \qquad \qquad let $\tilde{n}_k= l_i ^2$ 
                
                 \qquad   \qquad 
                let $\tilde{I_2} = \tilde{I_2}  \cup \{ \langle \tilde{n}_k, \tilde{p}_k\rangle\}$ 
                 
                 \qquad   \qquad 
                mark $I_k$
                
    \item if $I_k$ is still unmarked,   $gs_k = t$        
      
\end{enumerate} 
\item $j = j + 1$

\end{enumerate}
\item   $ \tilde{I} = \tilde{I_1} \cup \tilde{I_2} $
\end{enumerate}
\end{enumerate}

The analysis of the algorithm and its generated sketch $\tilde{I}$ is quite involved.  We proceed as follows. First, we  show that, with high probability,  $w_0$ provides a good bound in the sense that the total processing time of all jobs with processing times greater than $w_0$ has small impact on estimating $OPT$.   Then we show that 
the number of intervals in $H$ is bounded by  $h = O(\log m n)$. For all the intervals in $H_1$ and $H_2$, we will show that the algorithm computes a good estimate of the number of jobs in each interval of $H_1$ and $H_2$. Consequently, we can get a good estimation of $OPT$ through the sketch instance $\tilde{I}$. Finally, we show that the total number of samples drawn by the algorithm is sublinear. The detailed analysis is presented subsequent subsections. 

\subsubsection{  {$ \sum_{p_j > w_0 } p_j  \le \delta OPT$} and  {$h = O(\log {m n})$} }
 
\begin{lemma}\label{w-0b-enough-lemma} Let $w_{\delta}$ be the smallest processing time among all jobs  such that $\sum_{p_j>w_{\delta}} p_j\le \tfrac{\delta P}{ m}$. 
Then, with probability at least $1-\gamma_0$, $w_0 \ge w_{\delta}$, which implies   
$ \sum_{p_j > w_0 } p_j  \le \delta OPT$.  
\end{lemma}

\begin{proof} Assume that  $w_0<w_{\delta}$. Then by the definition of $w_{\delta}$,
  we  have $ \sum_{p_j  > w_0 } p_j   
  >   \tfrac{\delta P}{ m}$. 
Therefore,  the probability of sampling a job with   processing time less than or equal to $w_0$ is at most 
$(1-\tfrac{\delta}{m})$.  Since $w_0$ is the largest processing time among all $K$ samples,   all these  samples   have a processing time at most  $w_0$.     
However, the probability of this happening is at most $(1-\tfrac{\delta}{m})^{{K}}$. By the choice of ${K}$ and $d_0$ in step~\ref{line-d0} of the algorithm,   $$(1-\tfrac{\delta}{m})^{{K}} = (1-\tfrac{\delta}{m})^{\tfrac{50 m}{\delta} d_0} < (1-\tfrac{\delta}{m})^{\tfrac{m}{\delta} d_0}\le  \gamma_0 \enspace.$$

Thus, with probability of at least $1-\gamma_0$, we have  $ w_0 \ge w_{\delta}$. Consequently, $\sum_{p_j > w_0} p_j \le \sum_{p_j > w_{\delta}} p_j \le  \tfrac{\delta P}{ m}  \le \delta OPT$, since $\tfrac{P}{m} < OPT$. 
\end{proof}

\begin{lemma}\label{h-bound-lemma}   
With probability at most $ \gamma_0$, we have 
$h > \tfrac{1}{\delta}\ln{\tfrac{8nm}{ \delta (1-\delta)}}$. 
\end{lemma}

\begin{proof} 
Let $h_1$ be the largest index of the processing time intervals such that $( \sum_{j  \in I_t , t \ge h_1} p_j)  > \tfrac{\delta P}{ 8 m} $. We first show that $h_1  < \tfrac{1}{\delta}\ln{\tfrac{8nm}{\delta (1-\delta)}}$. 

If a job $j  \in I_t =(\tilde{p}_{t+1}, \tilde{p}_t]$  where $ t \ge h_1 $, then $p_j \le \tilde{p}_t \le \tilde{p}_{h_1}  = (1 - \delta)^{h_1-1} w_0$. Therefore, we can bound the total processing time of the jobs in the intervals $I_t$, $t \ge h_1$, $$  \sum_{j  \in I_t , t \ge  h_1} p_j   \le n \cdot (1 - \delta) ^{h_1 -1} w_0 \le       \tfrac{n}{(1 - \delta)} \cdot e ^{ - \delta h_1} P \enspace.$$ By  the definition of $h_1$, we must have  
$$  \tfrac{\delta P}{8m} <    \sum_{j  \in I_t , t \ge  h_1} p_j   \le       \tfrac{n}{(1 - \delta)} e ^{ - \delta h_1} \cdot P \enspace, $$
from which we get $h_1  < \tfrac{1}{\delta}\ln \tfrac {8nm}{ \delta (1 - \delta)}$. 

Now, assume $h > h_1 $. Then, by definition of $h_1$,  we must  have $  \sum_{p_j\in I_t, t\ge h} p_j  \le \tfrac{\delta P}{ 8m}$. Applying Chernoff bound with $\alpha =1$ and $p = \tfrac{\delta }{ 8m}$, we get, with probability at most $$e^{-\tfrac{1}{3}(1)^2 \tfrac{\delta}{8m}{K}} \le e^{-\tfrac{1}{3}(1)^2 \tfrac{\delta}{8m}\tfrac{50m}{\delta} d_0} \le  e^{-2 d_0} <   (1-\tfrac{\delta}{m})^{\tfrac{m}{\delta} d_0 } \le  \gamma_0 \enspace,$$   the number of samples from intervals $I_t$,  $t\ge h$,  is  more than  $(1+1)\cdot  \tfrac{\delta }{ 8 m}{K} = \tfrac{\delta }{ 4 m}{K}$. 
On the other hand, by Step~
\ref{h-assignment-line} of the   algorithm,  
 $h$ is the smallest integer  such that $\sum_{k > h} X_k \le \tfrac{\delta {K}}{ 4m}$, which implies that  
the number of samples from $I_t$ for all   $t\ge h $, including the interval $I_h$ must be more than  $\tfrac{\delta {K}}{ 4m}$. Thus, the probability that   $h > h_1 > \tfrac{1}{\delta}\ln{\tfrac{8nm}{ \delta (1-\delta)}} $ is at most    $ \gamma_0$.
\end{proof}

 \subsubsection{Intervals in $H_1$: $\tilde{n}_k$ as  an Approximation of  ${n}_k$}

Next, we consider the intervals in $H_1$ and show that $\tilde{n}_k$  provides an approximation of ${n}_k$. We  first present a  technical lemma that is used in the proofs of several other lemmas.

\begin{lemma}\label{simple-ineqn-lemma}
For each $k\ge K_0$, $\beta\ge {\beta_0}$, and $a\ge e^{\tfrac{1}{8}}$, we have $\tfrac{1}{\beta}a^{-\beta k}\le \tfrac{\gamma_0}{ e}$, where $K_0$ and $\beta_0$ are defined in Adaptive-Sketch-Algorithm.  
\end{lemma}

\begin{proof}
It is obvious $\tfrac{1}{\beta}a^{-\beta k}  
 \le  \tfrac{1}{\beta_0} (e^{\tfrac{1}{8}})^{-\beta_0 K_0}$, and by the definition of $K_0$ at Step~\ref{adaptive-k0-line} of the algorithm, we have 
$$ 
\tfrac{1}{\beta}a^{-\beta k}  \le  \tfrac{1}{\beta_0}(e^{\tfrac{1}{8}})^{-\beta_0 K_0} 
 = \tfrac{1}{\beta_0}(e^{\tfrac{1}{8}})^{-\beta_0 \cdot \tfrac{8 }{ \beta_0}\cdot \ln \tfrac{e }{ \beta_0 \gamma_0}} 
 = \tfrac{\gamma_0}{ e} \enspace.$$  
\end{proof}

Lemma~\ref{small-number-interval-lemma0} and Lemma~\ref{small-number-interval-lemma} are applications of Chernoff bound. These lemmas demonstrate that when we draw samples using weighted sampling, if a job has a large processing time, it is unlikely to be sampled only a few times; conversely, if a job has a small processing time, it is unlikely to be sampled many times.

\begin{lemma}\label{small-number-interval-lemma0}
Assume that $k$ samples are drawn. The  probability that there exists a job $j$ such that   $p_j\ge \beta_1 P $ but $ x_j < (1-\alpha)\beta_1 k$,  is at most $\tfrac{1}{\beta_1} e^{-{\tfrac{1}{ 2}}\alpha^2 \beta_1 k}$.    Furthermore, if $k \ge  K_0$, $\alpha\ge \tfrac{1}{2}$ and $\beta_1 =  \beta_0 $, the probability is at most $\gamma_0$.
\end{lemma}

\begin{proof} 
Let $j$  be a job with $p_j\ge \beta_1 P $. If we draw a single sample,   the probability  of selecting job $j$ 
is at least $\beta_1$. When drawing $k$ samples,   by Chernoff bound,  the probability that the number of occurrences of job $j$   is less than $(1 - \alpha) \beta_1 k$ is
$Pr(x_j < (1-\alpha)\beta_1 k)  \le e^{-\tfrac{1}{2}\alpha^2 \beta_1 k}$. 

Given that there can be at most $\tfrac{1 }{\beta_1}$ jobs with processing time at least $ \beta_1 P$, 
by union bound, the probability that there exists at least one job $j$ with $p_j\ge \beta_1 P$ but  $x_j < (1-\alpha)\beta_1 k$, is  at most 
 $v_1=\tfrac{1}{\beta_1} e^{-\tfrac{1}{2}\alpha^2 \beta_1 k} = \tfrac{1}{\beta_1} (e^{\tfrac{1 }{ 2} \alpha^2 })^ {-\beta_1 k} $. 

If  $\alpha\ge \tfrac{1}{2}$, then $e^{ \tfrac{1}{2}\alpha^2}  \ge e^{
\tfrac{1}{8}}$. In addition, if    $k \ge K_0$,
and $\beta_1= \beta_0$, by Lemma~\ref{simple-ineqn-lemma}, we have 
$$v_1=\tfrac{1}{\beta_1} (e^{\tfrac{1 }{ 2} \alpha^2} )^ {-\beta_1 k}  \le \tfrac{1}{\beta_1} e^{-\tfrac{1}{8}  {\beta_1 k}} 
  \le \tfrac{\gamma_0}{ e}\le \gamma_0.$$
\end{proof}

\begin{lemma}\label{small-number-interval-lemma}Let $k$ be an integer, and $\beta_2\in (0,1)$ such that $\beta_2 k\ge 1$.
Assume that $k$  samples are drawn. Then the probability that 
there exists a job $j$  such that $p_j\le \beta_2 P$ and  $ x_j \ge 4\beta_2 k$, is at most $\tfrac{e }{ \beta_2}\cdot (\tfrac{e }{
4})^{\beta_2 k}$.  Furthermore, if $k \ge K_0$,   and $\beta_2= 2 \beta_0 $, then the probability is at most $\gamma_0$.   
\end{lemma}

\begin{proof} 
We examine the jobs based on the processing time intervals that they belong to.
Consider jobs $j$ such that $ \tfrac{1}{2^{i}} \beta_2 P < p_j \le \tfrac{1}{2^{i-1}} \beta_2 P$, $i \ge 1$. There are at most $\tfrac {2^{i}} {\beta_2} $ such jobs.  
If we draw one sample,   the probability  that   a particular job $j$   is sampled is at most $\tfrac{1}{2^{i-1}} \beta_2 $. Suppose that  we draw $k$ samples 
and use the following Chernoff bound for any $\alpha > 0$ when the sample is drawn with probability $p$ : 
\begin{equation}\label{ineq:x_i}   
Pr(x_j > (1+\alpha)p k)  < \left[ \tfrac{e^{\alpha}}{
(1+\alpha)^{(1+\alpha)}}\right]^{ p k} < 
\left[\tfrac{e }{
(1+\alpha)}\right]^{\alpha \cdot p k} \enspace .\end{equation} 

 Adding in inequality  (\ref {ineq:x_i})    with   $\alpha  = 2^{i+1} -1$, and $p = \tfrac{1}{2^{i-1}} \beta_2 $,   we get 
$$Pr(x_j > (1+\alpha) p k) = Pr(x_j > 4 \beta_2 k)  
< \left (\tfrac{e}{2^{i+1}} \right)^{(2^{i+1}-1)  \cdot \tfrac{1}{2^{i-1}} \beta_2 k}   < \left (\tfrac{e}{2^{i+1}} \right)^{  2 \beta_2 k}.$$ Using the union bound, the probability that there exists at least one such job in this interval that is sampled more than $4 \beta_2 k$ times is at most
\begin{eqnarray*}  
\tfrac {2^{i}} {\beta_2} \cdot 
\left (\tfrac{e}{2^{i+1}} \right)^{  2 \beta_2 k} 
&\le&  
\tfrac {2^{i}} {\beta_2} \cdot \left (\tfrac{e}{2^{i+1}} \right)  \left (\tfrac{e}{2^{i+1}} \right)^{   \beta_2 k} \\
&\le&\tfrac{e }{ 2 \beta_2}\cdot  \left (\tfrac{e}{2^{i+1}} \right)^{    \beta_2 k} \\
& = &\tfrac{e }{ 2 \beta_2}\cdot  \left (\tfrac{e}{4}\right)^{    \beta_2 k} \cdot (\tfrac{1}{2^{i-1}})^{    \beta_2 k} 
\end{eqnarray*} 

Apply union bound to all processing time intervals  $ (\tfrac{1}{2^{i}} \beta_2 P,  \tfrac{1}{2^{i-1}} \beta_2 P] $, $i \ge 1$,  
we get the probability  that there exists at least one   job whose processing time is less than $\beta_2 P$ but sampled more than $4 \beta_2 k$ ($\beta_2 k \ge 1$ ) times is at most
$$ v_2 = \sum_{i=1}^{+\infty}\tfrac{e }{ 2 \beta_2}\cdot  \left (\tfrac{e}{4}\right)^{    \beta_2 k} \cdot (\tfrac{1}{2^{i-1}})^{    \beta_2 k} \le\tfrac{e }{ 2 \beta_2}\cdot  \left (\tfrac{e}{4}\right)^{    \beta_2 k} \cdot 2 =\tfrac{e }{ \beta_2}\cdot  \left (\tfrac{4}{e}\right)^{-    \beta_2 k} \enspace.$$
 
 Since $\beta_2=   2 \beta_0 $, and $\tfrac{4}{e} > e^{\tfrac{1}{8} }$, if $k \ge K_0$ (see Step~(\ref{adaptive-k0-line}) of Adaptive-Sketch-Algorithm), by Lemma~\ref{simple-ineqn-lemma}, we have 
\begin{eqnarray*}
v_2&=&\tfrac{e }{ \beta_2}\cdot \left(\tfrac{4 }{
e}\right)^{-\beta_2 k}= e\cdot \left(\tfrac{1}{\beta_2}\cdot \left(\tfrac{4 }{
e}\right)^{-\beta_2 k}\right)\le e\cdot  \tfrac{\gamma_0}{ e}\le \gamma_0.
\end{eqnarray*}
\end{proof}

With the previous two lemmas, we can show that for each interval in $H_1$, with high probability,  all jobs  from the interval are sampled, allowing us to get the exact number of jobs in it. 

\begin{lemma}\label{shared-Approx-lemma}  
With probability at least $(1-2\gamma_0)$, 
$\tilde{n}_k = {n}_k$ for all $k$, $I_k \in H_1$.
\end{lemma}

\begin{proof}
By construction, for each $\langle \tilde{n}_k, \tilde{p}_k\rangle \in \tilde {I_1} $,   there exists a job $j \in  I_k$ such  that
$ x_j \ge 8 \beta_0 K_0 $.
By Lemma~\ref{small-number-interval-lemma},  the probability that  $p_j \ge 2 \beta_0 P $ is at least $1 - \gamma_0$.  

Next, assume that $p_j \ge 2 \beta_0 P$. We show that, with high probability, all jobs in any interval containing at least one  of those jobs $j$ are also sampled. 
By definition of intervals,  if a job $r$ is in the same interval as job $j$, $p_r \ge p_j   (1 - \delta)$. Since   $\delta < \tfrac{1}{2}$,   if $p_j \ge 2 \beta_0 P $, then $$p_r \ge (1 - \delta) p_j \ge p_j /2 \ge \beta_0 P .$$ 
When $k$, $k > K_0$, samples are drawn, by Lemma~\ref{small-number-interval-lemma0}, the probability that there exists a job  $r$  such that   $ p_r \ge \beta_0 P$ and $x_r < \beta_0 k/2   $ is at most $ \gamma_0$.  
 
In conclusion,  the probability that there exists a $k$ such that  $\langle \tilde{n}_k, \tilde{p}_k\rangle \in \tilde {I_1} $ and   $\tilde{n}_k \neq  n_k$ is at most $2 \gamma_0$. 
\end{proof}

 \subsubsection{Intervals in $H_2$: $\tilde{n}_k$ as an Approximation of ${n}_k$ }

Intuitively, an interval from $H_2$ does not contain jobs with processing times large enough to be sampled  as frequently as those from intervals in $H_1$, but it contains a large number  of jobs. We show that our algorithm  provides an approximation for the number of jobs in each $I_k\in H_2$ based on the principal of birthday-paradox. 

For a particular $I_k \in H_2$, let $F_{i,k}$ be the event that,  at the end of an iteration  of Step~
\ref{h2-nk-single-round}, either $l_i <  (1 + \delta)^{-{c}} \cdot \sqrt {n_k}$ and   $I_k$ is marked, or  $ l_i \ge (1 + \delta)^{{c}} \cdot \sqrt{n_k}$ and $I_k$ remains unmarked. 

\begin{lemma}\label{shared-Approx-lemma2}
For a given $I_k \in H_2$ and $i \ge 1$,  $Pr(F_{i,k}) < \gamma_i$
\end{lemma}
\begin{proof}
First, we consider the probability of   $I_k$ is marked at the end of an iteration Step~
\ref{h2-nk-single-round}, where  $l_i   \le   (1 + \delta)^{-{c}} \cdot \sqrt{n_k}$. 
For a specific group, 
the probability that the first $l_i$ samples are all distinct in a group $G_j$ is at least
$Pr_{lb}(n_k, l_i, \delta) \ge  Pr_{lb}(n_k, \tfrac{\sqrt{n_k}}{{{(1+\delta)^{{c}}}}}, \delta) $.    Let $p_L=
  Pr_{lb}(n_k, \tfrac{\sqrt{n_k}}{{{(1+\delta)^{{c}}}}}, \delta)$.    
By the inequality (\ref{P_L}) from the preliminary section,  we have  $ p_L \ge \tfrac{1}{ (1-\delta) \sqrt{e} }  $. Thus
$$Pr( g_{l_i,k} \le \tfrac{1}{\sqrt{e}} u_i )   \le Pr(g_{l_i,k} \le (1-\delta) p_L u_i).$$ 
Applying Chernoff bound with $p=p_L$ to all $u_i$ groups, we get
$$Pr( g_{l_i,k} \le \tfrac{1}{\sqrt{e}} u_i )  
 \le Pr(g_{l_i,k} < (1 -\delta) p_L u_i )      \le  e^{-\delta^2 p_L u_i/2}  \le  e^{-\delta^2 p_L u_i/3} 
\enspace.$$
Note that $p_L   \ge \tfrac{1}{\sqrt{e}} $, 
and  by the definition of 
$u_i$  in Step~
\ref{h2-nk-single-round} of the algorithm, we have 
$$Pr( g_{l_i,k} \le \tfrac{1}{\sqrt{e}} u_i )   \le   e^{-\delta^2 p_L u_i/3} \le   e^{-\delta^2 \cdot \tfrac{1}{\sqrt{e}} (-\tfrac{ 3 {e} \ln \gamma_i}{\delta^2})\tfrac{1}{3}}\le e^{\tfrac{1}{\sqrt{e}} e \ln {\gamma_i}  }  \le \gamma_i \enspace.$$

 Now, we consider the probability that $I_k$ is not marked in iteration $i$  and $l_i > {(1+\delta)^{{c}}}{\sqrt{n_k}} $. The probability that the first $l_i$ samples are all distinct in a group $G_j$ is 
at most $P_u(n_k, l_i, \delta) \le P_u(n_k, (1+\delta)^{{c}}\sqrt{n_k}, \delta)$. 
Let $p_R =P_u(n_k, (1+\delta)^{{c}}\sqrt{n_k}, \delta) $. 
By inequality (\ref{P_R-upper-bound}) from the preliminary section,  $ p_R \le  \tfrac{1}{(1 + \delta )\sqrt{e}} $. Therefore, we get
$$ Pr( g_{l_i,k} \ge \tfrac{1}{\sqrt{e}}u_i ) \le Pr (g_{l_i,k} \ge (1+\delta)p_R u_i) \enspace.$$
Applying Chernoff bound with $p=p_R$ to all $u_i$ groups,
$$ Pr( g_{l_i,k} \ge \tfrac{1}{\sqrt{e}}u_i ) 
 \le Pr (g_{l_i,k} > (1 +\delta) p_R u_i )      \le  e^{-\delta^2 p_R u_i/3} 
\enspace. $$
By inequality (\ref{P_R-lower-bound}), we have $p_R \ge \tfrac{1}{e}$. And by the definition of $u_i$ in the algorithm,  we have 
$$ Pr( g_{l_i,k} \ge \tfrac{1}{\sqrt{e}}u_i ) 
  \le  e^{-\delta^2 p_R\cdot u_i/3} 
  \le  e^{-\delta^2\cdot \tfrac{1}{e}\cdot  ( -\tfrac{ 3 {e} \ln \gamma_i}{\delta^2})\cdot \tfrac{1}{3}}
\le  e^{\ln {\gamma_i}  }  \le \gamma_i \enspace.$$

%
%
%

Combining the above two cases, we proved the lemma. 
\end{proof}

\begin{lemma}\label{stop-prob-lemma}
The probability that there exists an interval $I_k$ such that $\langle \tilde{n}_k, \tilde{p}_k\rangle \in \tilde {I_2}$ and    $\tilde{n}_k \notin  [(1 + \delta)^{-2 {c}} \cdot n_k,  (1 + \delta)^{2 {c}} \cdot n_k)])$ where ${c} = 4 $,  is at most   $ \gamma_0$.
\end{lemma}
\begin{proof}  
For a particular interval $I_k \in \tilde {I_2}$, the probability that $\tilde{n}_k \notin  [(1 + \delta)^{-2 {c}} \cdot n_k,  (1 + \delta)^{2 {c}} \cdot n_k)])$ is $\sum_i Pr(F_{i,k} ) <  \sum_i \gamma_i$.   There are at most $h$ intervals in $H$. 
  The probability that there exists  one such interval is at most $h \sum_{i} \gamma_i  = h \sum_{i} \tfrac {\delta\gamma_0} {h l_i} = \delta\gamma_0\sum_{i\ge 1} \tfrac{1}{(1+\delta)^i}={\delta\gamma_0}\cdot \tfrac{1}{1+\delta}\cdot \tfrac{1}{1-\tfrac{1}{1+\delta}} 
  \le  \gamma_0$.
\end{proof}

\subsubsection{\texorpdfstring
  {$OPT(\tilde{I})$ Approximates $OPT(I)$}
  {OPT(I~) Approximates OPT(I)}
}
   In this subsection, we show that total processing time of the jobs from the  intervals excluded from $\tilde{I}$ is small, and that $OPT(\tilde{I})$ approximates $OPT(I)$. 

Let $I'$ be the instance corresponding to $\tilde{I}$, that is, $ \langle n_k, \tilde{p}_k\rangle   \in I'
 $  iff $ \langle \tilde{n}_k, \tilde{p}_k\rangle   \in \tilde{I}
 $.   Note that the difference between $I'$ and $\tilde{I}$ lies in the number of jobs only. With the  Lemmas~\ref{shared-Approx-lemma} and \ref{stop-prob-lemma}, we can use a  similar argument as in the proof of Lemma \ref{lemma:compress-error} to show that, with probability at least $(1-3\gamma_0)$, 
 \begin{equation}\label{eq:i'_i_tilde}
    (1 - 2 {c} m \delta ) OPT (I') \le OPT(\tilde{I}) \le (1 + 6 {c} m \delta ) OPT (I') \enspace,
\end{equation}   where ${c} = 4$ and $ \delta \le \tfrac{1}{4c}$.

Compared with instance $I$, the instance $I'$ excludes jobs from intervals not in $H$ and rounds up the processing times of the jobs from intervals in $H$.  In the following, we give an upper bound on the total processing time of those jobs that belong to intervals not in $H$.

\begin{lemma}\label{non-H-bound-lemma}
With probability at least $(1-3\gamma_0)$,
the total processing time of the jobs in the intervals $I_k$, $I_k\not\in H$, is at most $4\delta OPT(I)$.  
\end{lemma}

\begin{proof} There are three types of intervals that are not included in $H$. Consequently there are  three types of jobs not considered in the estimated sketch instance $\tilde{I}$. 

 The first type of jobs consists of  those whose processing time is greater than $w_0$.
By Lemma~\ref{w-0b-enough-lemma}, with probability at most  $\gamma_0$, the total processing time of these jobs exceeds ${\delta} OPT (I)$.

The second type of jobs includes those whose processing time is very small; specifically,  their processing times belong to intervals $I_k$ with $ k > h$.    We show that, with probability at most  $\gamma_0$, the total processing time of these jobs is more than   $ {\delta} OPT (I) $. 


  Let $h_2$ be the largest index of the intervals such that $( \sum_{j\in I_t,  t > h_2} p_j ) > \tfrac{ \delta P}{2 m} $. By Chernoff bound (with $\alpha=\tfrac{1}{2}$), with probability at most $$e^{-\tfrac{1}{2}(\tfrac{1}{2})^2(\tfrac{\delta}{2m})K} \le e^{-\tfrac{1}{16}  \tfrac{\delta K }{m}  }  < e^{-3 d_0} \le (1-\tfrac{\delta}{m})^{\tfrac{m}{\delta} d_0} \le   \gamma_0 \enspace,$$  we have $( \sum_{   t > h_2} X_t ) \le \tfrac{\delta {K}}{4 m}$.   By the definition of $h$ in Step~
  \ref{h-assignment-line} of Adaptive-Sketch-Algorithm,  we have $\sum_{k > h} X_k \le \tfrac{\delta {K}}{ 4m}$. That is, the probability that $( \sum_{j\in I_k,  k > h} p_j )$ is greater than    $\tfrac{\delta P}{2m}$ is at most $\gamma_0$. Consequently, with probability at most  $\gamma_0$, the total processing times of these jobs is greater $\delta OPT(I)$.

Finally, the last type of jobs not included in $H$ are those whose processing time belongs to an interval $I_k$, where $1 < k \le h$ but $X_k < \tfrac{\delta}{m h}  K_0$.  We can use  Chernoff bound (with $\alpha=\tfrac{1}{2}$) to show that for an interval $I_k$ whose jobs have  total processing time  exceeding $\tfrac{2\delta}{mh}P$ (thus $\tfrac{2\delta}{h} OPT$), 
the probability that $X_k< \tfrac{\delta }{ mh}K_0$   is at most   
\begin{equation}\label{inequality:large-P-small-Xk}
\begin{aligned}
e^{-\tfrac{1}{2}(\tfrac{1}{2})^2 \tfrac{2\delta}{mh} K_0}
&= e^{- \tfrac{\delta}{8mh}\cdot \tfrac{16}{\beta_0}\ln \tfrac{e}{\beta_0\gamma_0}}
 = e^{-\tfrac{64}{\delta^2}\ln \tfrac{e}{\beta_0\gamma_0}}
 \le e^{-\ln \tfrac{e}{\beta_0\gamma_0}} \\
&= \tfrac{\beta_0\gamma_0}{e}
 \le \tfrac{\delta^3\gamma_0}{32emh}
 \le \tfrac{\gamma_0}{32h}\,.
\end{aligned}
\end{equation}
 %
By union bound, 
with probability at most $h \cdot \tfrac{\gamma_0}{32h} = \tfrac{\gamma_0}{32}$, 
the total processing time of jobs from these intervals    is  greater than $2 \delta OPT(I)$. 

To sum up, the total processing time of the jobs from intervals not in $H$ is at most $4 \delta OPT(I)$. 
\end{proof}

\begin{theorem}\label{theorem: original-vs-compressed-2}
With probability at least $1-8\gamma_0$, 
$$(1 - 2 {c} m \delta) (1 - 4 \delta) OPT(I)  \le OPT(\tilde{I})  \le (1 + 6{c} m\delta) (1 + \delta)  OPT(I) \enspace, $$ where ${c} = 4$. 
\end{theorem}
\begin{proof}
   Since the instance  $I'$ excludes those jobs from intervals not in $H$,  by Lemma~\ref{non-H-bound-lemma}, we have $OPT(I)  \le OPT(I') + 4 \delta OPT(I) $. 
 Another difference between $I$ and $I'$ is that the processing times of the jobs in $I'$ are rounded up. This may increase the optimal makespan by at most a factor of $(1+\delta)$. Therefore, we have 
\begin{equation*}  (1 - 4 \delta) OPT(I) \le  OPT (I') \le  (1 + \delta) OPT(I). \end{equation*}
Together with inequality (\ref{eq:i'_i_tilde}), we have  
$$(1 - 2 {c} m \delta) (1 - 4 \delta) OPT(I)  \le OPT(\tilde{I})  \le (1 + 6{c} m\delta) (1 + \delta)  OPT(I) \enspace.$$
By union bound, the probability follows  Lemmas~\ref{w-0b-enough-lemma},  \ref{h-bound-lemma},  \ref{shared-Approx-lemma},  \ref{stop-prob-lemma}  and \ref{non-H-bound-lemma}.
\end{proof}

\subsubsection{The Time Complexity  of the Adaptive-Sketch Algorithm }

 The running time of the adaptive sketch algorithm is roughly proportional to the number of samples. In the following, we give an estimate of the total number of samples drawn by the algorithm.
 
\begin{lemma}\label{total-samples-lemma}With probability at least $(1-2\gamma_0)$,
     Adaptive-Sketch-Algorithm draws 
     $\tilde{O}(\tfrac{m}{\delta^4}\sqrt{n})$
     weighted random samples.
\end{lemma}

\begin{proof}  
Let $K_{max}$ denote the value of $K_j$ in the last iteration where all intervals are finally marked.  Since the number of samples that the algorithm draws in each iteration, $K_j$, grows geometrically, it is sufficient to  bound $K_{max}$.   

For each interval $I_k \in H$,  let $l'_{k}$  be the largest  $l_i = (1 + \delta)^i$ such that $ l_{i} < (1 + \delta)^{ {c}} \cdot\sqrt{ n_k}$, let $u'_{k} =  \lceil   \tfrac{ 3 {e} } {\delta^2} {\ln 
 \tfrac{h l'_{k}}{\delta \gamma_0} }\rceil  $ be the corresponding $u_i$. 
 Let $l'$ be the largest $l'_{k}$ for all $I_k\in H$ and $u'$ be the corresponding $u'_{k}$.
  In the following, we show   that  $ l' u'  X \le K_{max} < 2 l' u'  X $, where $ X = \tfrac{8mh}{\delta}\cdot \ln \tfrac{8h}{\gamma_0} $. 
 Specifically, we show that if $ K_j \ge l' u'  X $, then the number of samples $X_k$ from each unmarked interval $I_k$ can be divided into $u'$ groups, with each group containing   $l' \ge l'_{k}$ samples. 
 By Lemmas \ref{shared-Approx-lemma2} and \ref{stop-prob-lemma},  all intervals will be marked. 
 
 First, 
 we claim that, with  probability at most $\tfrac{\gamma_0}{32} $,  there exists an interval  $I_k \in H$ such that the  total processing time of the jobs in $I_k$  is less than  $\tfrac{\delta P}{2m h}$.
Consider  a particular interval $I_k$ whose jobs have  total processing time  not exceeding $\tfrac{\delta}{2mh}P$, 
the probability that  $X_k\ge  \tfrac{\delta }{ mh}K_0$   is at most   $$ e^{-\tfrac{1}{3}(1)^2  \tfrac{\delta}{2mh} K_0} < e^{-\tfrac{1}{8}  \tfrac{\delta}{mh} K_0}, $$ which is at most $\tfrac{\gamma_0}{32 h}$ by inequality (\ref{inequality:large-P-small-Xk}).
Applying the union bound to all  intervals (at most $h$), the claim follows. 

Secondly, when $K_j$ samples are drawn, using the Chernoff bound and union bound, one can show  that the probability  that there exists an interval $I_k$ such that the  total processing time of the jobs in $I_k$  is at least  $\tfrac{\delta P}{2m h}$  but the number of samples from $I_k$ (i.e., $X_k$)  is at most 
$\tfrac{\delta}{4mh}K_j$,  is  at most $h \cdot e^{-\tfrac{1}{3}\cdot (\tfrac{1}{2})^2 (\tfrac{\delta}{2m h})K_j}\le \tfrac{\gamma_0}{8}$.

Now, assume that $ X_k>  \tfrac{\delta}{4mh} K_j$. If $ K_j \ge l' u'  X$, then for each interval $I_k$, we have $$ X_k>   \tfrac{\delta}{4mh} K_j\ge \tfrac{\delta}{4mh} {l'u'X} = \tfrac{\delta}{4mh} {l'u'} \cdot \tfrac{8mh}{\delta}\cdot \ln \tfrac{8h}{\gamma_0} \ge l' u'\ge l'_{k} u'_{k}.$$ By Lemma~\ref{stop-prob-lemma},   the probability there exists an  interval that is not  marked is at most  $\gamma_0$. 

Using union bound, we get with  probability at least  $1 - \tfrac{\gamma_0}{32}  - \tfrac{\gamma_0}{8} - \gamma_0 \ge  (1 - 2 \gamma_0) $, we have $K_j \le    2 l'u'\cdot \tfrac{8mh}{\delta}  \ln \tfrac{8h}{\gamma_0}$ for all $K_j$ used in the algorithm. 
Note that $h \le  \tfrac{1}{ \delta} \ln \tfrac {8nm}{(1 - \delta)\delta} $, $l' \le (1 + \delta)^{ {c}+1}\sqrt{n} $ and $u' =(\lceil \tfrac{ 3 {e} }{\delta^2} (\ln\tfrac{h l'}{\delta \gamma_0}) \rceil) = O (\tfrac{1}{\delta^2}) \ln \tfrac{n}{\delta \gamma_0}$. Thus, we have
\[\begin{aligned}
K_j &\le 2 l'u' \cdot \tfrac{8mh}{\delta}  \ln \tfrac{8h}{\gamma_0} \\
    &= O\left(\tfrac{m}{\delta^4} \sqrt{n} \left( \ln \tfrac{n}{\delta \gamma_0} \right) 
    \left( \ln \tfrac{mn}{(1-\delta)\delta} \right)  \left( \ln \left( \tfrac{8}{\gamma_0} \ln \tfrac{8nm}{(1 - \delta)\delta} \right) \right) \right),
\end{aligned}
\]
 which means $K_{max} = \tilde{O}(\tfrac{m}{\delta^4}\sqrt{n})$.
\end{proof}

To summarize, we have shown  that, with high probability, the Adaptive-Sketch-Algorithm computes an $(\alpha, \beta_1, \beta_2)$ sketch where  
$\alpha = (1 + \delta)^{2c+2}$ (Lemmas \ref
{shared-Approx-lemma} and \ref{stop-prob-lemma},  ),   $\beta_1 = (1 + 6 {c} m\delta) (1 + \delta)$ (Theorem \ref{theorem: original-vs-compressed-2}), $\beta_2 = 4 \delta$ (Lemma  \ref{non-H-bound-lemma}). The size of the sketch instance is $O(h) = O(\tfrac{1}{\delta} \ln {\tfrac{nm}{(1-\delta)\delta}}) $.

\medskip

\begin{repeatedtheorem} {\ref{theorem:sublinear-n-unknown}}   Let $A(n, 1+\epsilon)$ be the time complexity for $(1+\epsilon)$-approximation scheme for makespan scheduling of $n$ input jobs. 
When the number of jobs is unknown, there is a randomized $(1+ \epsilon)$-approximation algorithm for the makespan scheduling problem  that uses adaptive weighted random samplings,   runs in $\tilde{O}\left( \tfrac{m^5} {\epsilon^4} \sqrt{n}    +  A(\ceiling{\tfrac{m} {\epsilon}}, m, {\epsilon} ) \right)   $
 time. 
Furthermore, it can compute  a sketch schedule, represented using   $O(\tfrac{m^2}{\epsilon^2} (\log {\tfrac{nm}{ \epsilon}))}$ space, which can be used later to generate a schedule with a makespan at most $(1+ 3\epsilon)$ times the optimal.  
\end{repeatedtheorem}
\medskip

\begin{proof}
Given instance $I$ with unknown number of jobs $n$, and parameters $\epsilon$ and $\gamma_0$,  we first  invoke the Adaptive-Sketch-Algorithm to get the sketch instance $\tilde{I}$ and then apply the Meta-algorithm  to  $\tilde{I}$   with appropriate parameters. For example, let  $\delta = \tfrac{\epsilon} {12 {c} m}$, and $\alpha  = \tfrac{\epsilon}{4}$, 
by Theorem~\ref{theorem: original-vs-compressed-2}, Lemma~\ref{total-samples-lemma}, and Theorem~\ref{trans-thm},   we can  get a $(1+ \epsilon)$-approximation   with probability at least $1-8\gamma_0-2\gamma_0=1-10\gamma_0$. 
Furthermore, by Theorem~\ref{meta-sketch}, it can generate a sketch schedule whose size is  $ O(m h) =  O(\tfrac{m^2}{ \epsilon^2} (\log \tfrac{n m}{\epsilon}))$ 
which  can be used to generate a schedule whose makespan is at most $(1+3\epsilon) OPT(I)$.

 The running time of the algorithm consists of two parts. The first computes the sketch instance using Adaptive-Sketch-Algorithm, which takes $\tilde{O}(\tfrac{m}{\delta^4}\sqrt{n})$ time. Plugging in $\delta = \tfrac{\epsilon} {12 {c} m}$, we get   $\tilde{O}\left( \tfrac{m^5} {\epsilon^4} \sqrt{n} \right)$ time. 
The second part is for finding approximation of the sketch instance and the time complexity  follows from Theorem~\ref{trans-thm}.
\end{proof}

\section{Application to Deterministic Algorithms Design}\label{Discussion:Deterministic-Algorithm-Sec} 
We show that by incorporating the idea of sketch instances and existing deterministic approximation algorithms, we can obtain faster deterministic approximation algorithms.
 
The only source of randomness in our algorithms is in the first step, where the input sketch is constructed. To obtain a deterministic algorithm, all we need to change is  to construct the input sketch deterministically. This can be done by scanning the input in two rounds. In the first round, we find the largest processing time $p_{max}$.  Let $\delta=\tfrac{\epsilon }{ 8}$, $\tilde{p}_k=(1-\delta)^{k-1} p_{max}$ and $I_k=(\tilde{p}_{k+1}, \tilde{p}_k]$ for $k\ge 1$. In the second round, we find $n_k$, the exact number of jobs in interval $I_k$. Let $h$ be the smallest integer such that  $n \cdot \tilde{p}_{h+1}  \le \delta p_{max}$, then we have $h = \theta (\tfrac{1}{\delta} \log \tfrac{n}{\delta} )$. Let the input sketch   $\tilde{I} = \{<n_k, \tilde{p}_k>, 1 \le k \le h\}$, it is easy to see  that  $ (1- \delta) OPT(I) \le OPT(\tilde{I} )\le (1+\delta) OPT(I) $.  In particular, we can easily show that it is an $(\alpha, \beta_1, \beta_2)$-sketch, where $\alpha = 1$, $\beta_1 = 2 \delta = \tfrac{\epsilon}{4}$, and $\beta_2 = 2 \delta = \tfrac{\epsilon}{4}$. Thus, we have the  following lemma.

\begin{lemma}\label{deterministic-sketch-lemma}
    There is a $O(n)$ time algorithm such that given  an instance $I$ and $\epsilon\in (0,1)$, it outputs  an $(\alpha, \beta_1, \beta_2)$-sketch, $\tilde{I} = \{<n_k, \tilde{p}_k>, 1 \le k \le h\}$ where $\alpha = 1$, $\beta_1 = \beta_2 =  \tfrac{\epsilon}{4}$, and  $h=O(\tfrac{1}{\epsilon} \log {\tfrac{n}{ \epsilon})}$.
\end{lemma}

Once we get the input sketch, we can either apply the  Meta-Algorithm to get an approximation or use the Modified Meta-Algorithm to get a sketch schedule which then can be converted to a real schedule.  
\begin{theorem} \label{deterministic-theorem}
  Assume that there is a $(1+\alpha)$-approximation algorithm $A$ for the makespan problem of $n$ jobs on $m$ machines, with running time $A(n,m,\alpha)$. Then for any $\epsilon\in (0,1)$, there is a $O(A(h(m,\delta), m, \delta)+n)$  time $(1+\epsilon)$-approximation scheme for an input of $n$ jobs, where  $\delta =\tfrac{\epsilon}{3}$, and $h(m,\delta)=\ceiling{\tfrac{m}{\delta}}$. A schedule can be obtained in additional $O( \tfrac{1}{\epsilon} \log \tfrac{n}{ \epsilon} + m)$ time. 
\end{theorem}

Note that, from Theorem \ref{trans-thm}, when $A(\tfrac{3m}{\epsilon} , m, \tfrac{m}{\epsilon}) = O(n)$, we have an  approximation scheme that runs in linear time. For example,  for the algorithm  in ~\cite{JansenKleinVerschae20},  $A(\tfrac{3m}{\epsilon} , m, \tfrac{m}{\epsilon}) =  2^{O(\tfrac{1}{\epsilon}(\log \tfrac{1}{\epsilon})^4)} + \tfrac{m}{\epsilon} (\log \tfrac{m}{\epsilon})$ which is $O(n)$ when $m = O(\tfrac{n }{ \log n})$. 
Consequently, by invoking the algorithm of~\cite{JansenKleinVerschae20} only on a  subset of jobs in the sketch instance, we obtain a deterministic approximation scheme with overall running time 
$O(n)$. This is asymptotically faster than applying the algorithm of~\cite{JansenKleinVerschae20} directly to the entire instance.

\section{Conclusions}\label{sumary-sec}

In this paper, we designed  sublinear time algorithms for the classical parallel machine makespan minimization scheduling problem. Using proportional weighted random sampling, our algorithms return an approximate value for the optimal makespan 
in sublinear time. These algorithms will help manufactures and service industry to make fast decision on production planning as well as accepting/rejecting orders. Furthermore, our algorithms can generate a sketch schedule which, if needed, can be used to produce a real schedule when the complete job information is obtained. It is worth noting that weighted random sampling is well suited to many scheduling problems with arbitrary job processing times, thereby broadening the potential applications of our approach to other scheduling settings as well.

\bibliographystyle{abbrvnat}
\bibliography{bibliography.bib}
\label{sec:biblio}

\end{document}